\documentclass{article}
\usepackage{graphicx,caption}
\usepackage{pgf,tikz}
\usetikzlibrary{arrows}
\usepackage{array,multirow,makecell}
\setcellgapes{1pt} \makegapedcells
\usepackage{amssymb}
\usepackage{subfig}
\usepackage[normalem]{ ulem }
\usepackage{soul}
\usepackage{graphicx}
\usepackage{here}
\usepackage{epsfig}
\usepackage{color}
\usepackage{pgf,tikz,amsthm}
\usetikzlibrary{arrows}
\usepackage{amsmath}
\usepackage{epstopdf}

\usepackage{breqn}

\usepackage{hyperref}
\hypersetup{
    colorlinks=true,
    linkcolor=blue,
    filecolor=magenta,      
    urlcolor=cyan,
    citecolor=red
}
\usepackage{verbatim,cleveref}
\usepackage[]{geometry}
\newtheorem{theorem}{\textbf{Theorem}}

\theoremstyle{definition}

\newcommand{\N}{\mathcal{N}}

\newcommand{\Q}{\mathcal{Q}}

\usepackage{tikz}
\usepackage{anyfontsize}

\usepackage[backend=biber, bibstyle=numeric-comp, citestyle=numeric-comp, giveninits=true, sorting=nyt, maxcitenames=2,  sortcites=true, maxbibnames=99, isbn=false, url=false, doi=true, date=year, eprint=true, uniquelist=false, uniquename=false]{biblatex}

\title{Analysis of a two patch model for disease vector-animal dynamics with \blank{non-linear} anthropization-driven migration}
\author{O.W. Happi-Tchakount\'e$^{1}$,   I.V. Yatat-Djeumen$^{1,2,3}$,  L. Eigentler$^{4,5}$\footnote{Corresponding author: lukas.eigentler@warwick.ac.uk}, P. Couteron$^{3}$ \\	
\small $^1$University of Yaound\'e I, National Advanced School of Engineering of Yaound\'e, \\ 
	\small Department of Mathematics and Physics, Yaound\'e, Cameroon \\
\small $^2$AMAP, Univ Montpellier, IRD, CIRAD, CNRS, INRA Montpellier, France \\
\small $^3$IRD, Sorbonne Université, UMMISCO, F-93143, Bobby, France \\
\small $^4$ Warwick Mathematics Institute, University of Warwick, Coventry CV4 7AL, United Kingdom \\
\small $^5$ Zeeman Institute for Systems Biology \& Infectious Disease Epidemiology Research, \\ \small University of Warwick, Coventry CV4 7AL, United Kingdom 
}
\date{\today}

\newcommand{\blank}[1]{\textcolor{black}{#1}}

\begin{document}

\maketitle

\begin{abstract}
Landscape dynamics are key drivers of the movement and distribution of sylvatic \blank{hematophagous} disease vectors and their (wild) animal hosts. Their habitats are undergoing increasing change, particularly fragmentation, through anthropogenic activity. In this article, we present and analyse a novel mathematical model that explicitly combines anthropization-induced landscape dynamics with the population dynamics of hematophagous vectors and (wild) animals dynamics. We develop a phenomenological and analytically tractable two-patch model in which the migration terms between the patches \blank{nonlinearly} depend on the anthropization level of the patches. Our model analysis comprising analytical stability analysis and numerical bifurcation analysis provides information on how changes in model parameters, especially anthropization levels, shape the long-term dynamics in the model. Precisely, we find that low anthropogenic activity allows for a vector-animal coexistence state, while high anthropization leads to a vector extinction state. However, we establish that for intermediate anthropization levels, the transition between the two states is not necessarily monotonic, but may instead occur via a sequence of concurrent bifurcations along the anthropization axis.
\end{abstract}

\section{Introduction}

Vector-borne diseases (VBD) affect approximately one billion people and represent 17\% of all infectious diseases \cite{tibayrenc2024genetics}. The progression of transmissible diseases such as VBD, as described by Pavlovsky (see e.g., \cite{Emmanuel2011,Lambin2010}), hinges on the continuous interaction of five essential prerequisites: (1) the donor animals (or host, which may be humans), (2) the vectors involved, (3) the recipient animals, (4) the infectious pathogenic agent, and (5) external environmental influences that facilitate the spread of the pathogen \cite{Emmanuel2011,Lambin2010}. This is also referred to as the landscape epidemiology. The specific spatial arrangement and interactions of these five components within a given area directly influence the observed infection risk and may play a role in disease emergence. Indeed, the dynamics of landscapes (as a result of climate change, human action, etc.) significantly modify vectors and animal hosts habitats, influence the movement of vectors and animal hosts, which in turn affects the epidemiology of vector-borne diseases \cite{Crowder2013,Ferraguti2021,Fletcher2018,Lambin2010,Shaw2024,Tracey2014}. Understanding these interactions is crucial for predicting disease transmission patterns and managing public health  risks \cite{Patz2004}.

Vectors are pivotal in vector-borne disease transmission, acting as conduits that move pathogens either mechanically or biologically \cite{tibayrenc2024genetics, higley1989manual}. Mechanical transmission is when the vector acts as a passive carrier, simply transporting a pathogen on its body (e.g., legs, mouth) without being infected itself and transferring it to a new host through physical contact or by contaminating food, e.g. a fly landing on fresh food after coming into contact with faeces \cite{higley1989manual}. Biological transmission happens when the pathogen's development or replication cycle necessitates its presence inside the vector \cite{SAUCEDO2022742, higley1989manual}. A vector ingests the pathogen during feeding (blood meal); subsequently, the pathogen undergoes an incubation phase. Only once this phase is complete can the vector successfully infect a new organism, frequently by inoculation during a subsequent blood meal \cite{SAUCEDO2022742, higley1989manual}. 

Landscapes that are habitats of disease vectors and animal hosts (among others) are currently changing at an unprecedented rate across the globe due to anthropization as land use is correlated with the size of the human population \cite{Ellis2013,KathleenLyons2016}. Anthropization leads to biodiversity loss \cite{KathleenLyons2016}, because it causes changes to local microclimatic conditions \cite{Williams2020}. For example, habitat loss significantly reduced the population size and demographic health of the specialist herbivore \textit{Chelinidea vittiger} (a species of leaf-footed bug), with severe effects observed after a threshold of approximately $70\%-80\%$ patch loss \cite{Fletcher2018}. Moreover, forest canopy removal typically leads to warmer \cite{Alkama2016,Senior2017} and drier \cite{Frishkoff2016a} conditions and the intensification of temperature extremes (both hot and cold) \cite{DeFrenne2019,Alkama2016}.  The tropics, such as the Amazon or Sub-Saharan Africa, where vector-borne diseases are common, are particularly vulnerable to anthropization, with low levels in land-use changes leading to significant biodiversity loss \cite{Cantera2022}. Specifically, landscape transformation across Africa shows persistent expansion of cropland and settlements, large-scale conversion of woody and natural vegetation, and rising human appropriation of ecosystem productivity. Regional hotspots show rapid change over recent decades with measurable losses of natural cover and increases in built-up area and cropland \cite{Bullock2021-vt,Chiaka2021-hw,Fetzel2016-xc,Kimani2021-jb}. Indeed, Africa-wide human appropriation of net primary production reached about 20\% in 2005 and grew by roughly 55\% since 1980, with large regional differences in intensity and efficiency of biomass use \cite{Kimani2021-jb}. According to \cite{Bullock2021-vt}, cropland area in a seven-country East Africa sample increased by ~18,154,000 ha ($\approx$ 34.8\%) between 1998–2017, and Settlements area rose $\approx$43.5\% over the same period; woody classes were largely converted to less-woody uses with ~20 million hectares more conversion than succession toward woody recovery. In the same vein, \cite{Fetzel2016-xc} reported that natural cover around major mining agglomerations in southeastern Katanga (in Democratic republic of Congo) lost $>$60\% of its area between 1979 and 2020, with agricultural and built-up areas increasing sharply. Moreover, it was estimated that in Ivory Coast, dense forest and degraded forest in the south‑west declined at rates of about 1.44\%/yr and 3.44\%/yr respectively from 1987–2015, driven primarily by conversion to agriculture \cite{Chiaka2021-hw}.

Landscape structure, specifically attributes affected by anthropization such as fragmentation and connectivity, is widely recognized as a major factor influencing the complex interactions among hosts, vectors, and pathogens. For populations capable of range expansion, changes in land use leads to increased migration \cite{Frishkoff2016a}, which will be the main focus of this paper. Changes in habitat integrity can therefore alter disease spread; in particular, vector mobility is essential to the transmission cycle, as it disperses pathogens to previously unexposed hosts \cite{SAUCEDO2022742}. Specifically, local disease spread is facilitated by vector movement, while the relocation of infected hosts (such as humans) is responsible for long-distance dissemination of the disease.
Similarly, the interplay between landscape dynamics and mosquito ecology, as highlighted by \cite{Crowder2013} and \cite{Ferraguti2021}, critically affects the transmission dynamics of pathogens such as the West Nile virus and avian malaria.

The study of vector spatio-temporal dynamics has been the subject of several contributions that relied on mathematical models. The proposed modelling frameworks include reaction‑diffusion-advection partial differential equations (PDE), patch/metapopulation frameworks, stochastic and agent‑based simulations, 
see for instance \cite{Cailly2011,Cui2025-bt,daSilva2020,Dufourd2012-dd,Dufourd2013-uu, Dye2024-sw, Hancock2019,Heath2022, Lutambi2013, Lutambi2013MBS, Maneerat2016-th,McCormack2019,nguyen:tel-04687836, Roques2016} and references therein. Reaction-advection-diffusion models were proposed in \cite{Dufourd2012-dd,Dufourd2013-uu,dumont_spatio-temporal_2011} to model mosquito dispersal while taking into account environmental parameters, like wind, temperature, or landscape elements as well as a biological control. Specifically, the study by Dumont et al. \cite{dumont_spatio-temporal_2011} initially focused on a dispersal model for adult female mosquitoes, categorizing them into two classes: blood-meal searchers and breeding-site searchers. Subsequently, Dufourd et al. \cite{Dufourd2012-dd} expanded upon this foundation by incorporating the aquatic stage (larvae/pupae) and several adult mosquito classes: immature females, resting females, wild males, and sterile males. This comprehensive approach yielded a reaction-advection-diffusion model describing the mosquito population dynamics.
In their analysis, Dufourd et al. \cite{Dufourd2012-dd} introduced landscape heterogeneity by modeling it solely through the diffusion coefficient. They postulated that areas with lower diffusion rates represented favorable habitats where mosquitoes would exhibit prolonged residency (i.e., spend more time). Through numerical simulations of the resulting model, the authors demonstrated a significant impact of landscape structure on both mosquito spatial distribution and the efficacy of vector control strategies.
More recently, authors in \cite{Cui2025-bt,daSilva2020,,nguyen:tel-04687836, Roques2016} studied a reaction-diffusion model for mosquito population dynamics. Specifically, the model proposed in \cite{Cui2025-bt} accounts for seasonal variations in temperature across different geographical regions that directly influence the life cycle and reproduction rates of mosquitoes.

In \cite{Hancock2019}, authors develop spatially explicit models of Wolbachia-\textit{A. aegypti} dynamics that aims to incorporate realistic patterns of demographic variation in the mosquito population. Their  metapopulation model describes the spatially heterogeneous demography of \textit{A. aegypti} using empirically validated relationships that link density-dependent demographic traits to mosquito abundance. They showed that their model can produce rates of spatial spread of Wolbachia similar to those observed in field releases and can be applied to compare the performance of different strategies. In the same vein, metapopulation models for mosquito population dynamics are also proposed in \cite{Cailly2011, Dye2024-sw, Lutambi2013, Lutambi2013MBS}. Specifically,  \cite{Lutambi2013, Lutambi2013MBS} utilized a discrete-space, continuous-time (multi-patch) mathematical framework used to investigate how mosquito dispersal and the non-uniform distribution of resources (such as blood hosts and breeding sites) impacted the spatial distribution and overall population dynamics of the mosquitoes.	Their results clearly indicated that landscape heterogeneity and its temporal changes significantly affect both the distribution and the numerical abundance of the mosquito population.
Furthermore, \cite{McCormack2019} develop a stochastic metapopulation model of mosquito population dynamics and explore the impact of accounting for breeding site fragmentation when modeling fine-scale mosquito population dynamics. They found that when mosquito population densities are low, fragmentation can lead to a reduction in population size, with population persistence dependent on mosquito dispersal and features of the underlying landscape. The authors in \cite{Maneerat2016-th} proposed a spatially explicit behavioral agent-based simulation model of \textit{the} female mosquito \textit{Aedes aegypti} that examines the effects of neighborhood-scale variables (including human population density, the concentration of breeding sites and topological features) on the mosquito's population cycle. 

We also acknowledge a body of literature dedicated to the spatio-temporal analysis of vector-borne disease models. These studies often concurrently examine vector, human, and/or animal populations and the spread of the pathogen (e.g., \cite{Anzo-Hernandez2018-to,Arino2012-mp,Auger2008-lc,Esteva2019-mf,Gimenez-Mujica2025-sf,Iggidr2017-jj,Moulay2013-uu} and references therein).  We will not delve into these models as this paper solely considers population dynamics between vectors and their host, but does not model disease dynamics.

\blank{Crucially, an understanding of how anthropization-induced dispersal affects vector-animal dynamics is currently lacking. To this end, this paper presents a novel mathematical model that}  explicitly combines anthropization-induced landscape dynamics with hematophagous vectors and (wild) animals dynamics. We do so by introducing a new mathematical model of these dynamics on a network of two connected patches (\Cref{sec: model}). \blank{We specifically construct a phenomenological model that is analytically tractable to obtain qualitative insights into the impacts of anthropization on the population dynamics. We regard these as a necessary foundation for future work with more detailled models that, for example, consider more complex landscapes or couple the population dynamics to disease spread.} We exploit \blank{the analytical tractability of our model} by performing a stability analysis for the model's steady states in \Cref{sec: lin stab analysis} to reveal the impact of increased anthropization on the vector-animal dynamics. While exact, the derived stability conditions are 
\blank{complemented by} 
a numerical bifurcation analysis using numerical continuation  to gain better insights into these conditions (\Cref{sec: numcont}). Finally, \blank{we} discuss our results in \Cref{sec: discussion}.

\section{Model}\label{sec: model}
We start by introducing a new phenomenological model that describes the interactions between \blank{a hematophagous} vector population and their (wild) animal host on a network of connected patches affected by anthropization. In the interest of analytical tractability, we restrict the number of patches to two in this paper \blank{(see the discussion for a framework on $n \in \mathbb{N}$ patches).}

We consider a discrete spatial domain comprising two patches and we let $V_{i}$ and $A_{i}$ ($i= 1,2$) denote the \blank{hematophagous} vector population size in patch $i$ and the animal population size in patch $i$, respectively. \blank{Our objective is to understand the impact of anthropization, viewed through the lens of natural landscape fragmentation, on the dynamics of both hematophagous vectors and (wild) animal population. To model this, we introduce $\alpha_{i}\in[0,1)$ as a measure of anthropization/fragmentation pressure on the patch $i$. An $\alpha_{i}$ 
value of 0 implies no human-induced fragmentation, while values nearing 1 suggest almost complete anthropic fragmentation.}

We assume that in patch $i=1,2$, animals grow logistically with maximum growth rate $r_A>0$. In the absence of anthropization, animals have carrying capacity $K_{A,i}>0$. \blank{To maintain model simplicity, we \blank{assume} that substantial anthropic fragmentation of natural landscapes profoundly affects wild animal dynamics by \blank{both} diminishing their carrying capacity (see also \cite{dumont2025humanwildlifeinteractionstropicalforest} for a similar consideration) \blank{and triggering animal dispersal}. We also posit that the primary consequence of anthropization on  blood-feeding vectors is the stimulation of dispersal. Specifically, we \blank{assume} that in a pristine environment where the fragmentation level is zero ($\alpha_{i}=0$), we may expect no displacement of either vectors or animals. This stems from our baseline premise that every patch constitutes a viable habitat, removing the necessity for \blank{migration} when $\alpha_{i}=0$. Furthermore, we model the relationship between anthropic fragmentation and movement as a \blank{sigmoidal} response dictated by a critical threshold. That is, when $\alpha_{i}$ is low, displacement rates remain negligible whereas, when $\alpha_{i}$ is high, displacement rates for vectors and animals escalate significantly. Essentially, the impact of anthropic fragmentation on dispersal follows two distinct regimes—low and high—separated by a tipping point. We then model the feedback of $\alpha_{i}$ on the displacement rates by the functions
\begin{equation}\label{mAmV}
\begin{array}{l}
D_A(\alpha_{i})=\dfrac{\alpha_{i}^n}{\alpha_{i}^n+c_{A}^n},\\
D_V(\alpha_{i})=\dfrac{\alpha_{i}^n}{\alpha_{i}^n+c_{V}^n}
\end{array}
\end{equation}
where $n\blank{>1}$ controls the shape of the function. \blank{The parameter} $c_{A}>0$ (resp. $c_{V}>0$) is the animal-related (resp. vector-related) half saturation \blank{constant}.}

\blank{Moreover,} vectors grow logistically with maximum growth rate $r_V>0$. \blank{However, to \blank{facilitate larvae production}, hematophagous vectors ($V_i$) must obtain blood meals from wild animal hosts ($A_i$). Field observations and empirical data \cite{Likeufack2025,lord_host-Seeking_2017} indicate that tsetse fly females, for example, typically seek a blood meal every 2.5 days to complete their reproductive cycle. This biological requirement establishes a direct link between successful feeding and the vector's birth rate. To represent the beneficial impact of host availability on the vector population, we utilize a Monod-style functional response \cite{Tewa2013}:$$\dfrac{A_i}{A_i+a_i}$$ where $a_i>0$ represents the half-saturation constant. Indeed, while more animals ($A_i$) generally mean more larvae, there is a ``bottleneck'' where the vectors are limited by their own handling time or digestion speed, regardless of how many hosts are available. \blank{Further, we assume that} vector migration away from patch $i$ is decreasing with animal abundance in patch $i$, and the constant $b_i>0$ quantifies the strength of this dependence. The wild animals (resp. hematophagous vectors) maximum migration rate is $d_{A,ij}\ge 0$ (resp. $d_{V,ij}\ge 0$). In terms of migration, mass is conserved in the system, i.e., any mass emigrating from patch $i=1,2$ immigrates to patch $j\blank{\ne i}$. }

All population densities are assumed to decrease at constant rates ($\mu_{V,i}, \mu_{A,i} \ge 0)$ representing natural death. \blank{In the interest of generality, death rates, as well as the vector half saturation constants ($a_i$) may differ between patches. This accounts for differences between patches that are otherwise not explicitly described (i.e. properties other than anthropization).}
\blank{The resulting model with non-linear migration terms is}
\begin{equation}\label{2p-model}
    \left\{
\begin{array}{l}
     \dfrac{dV_{1}}{dt}= r_V\dfrac{A_{1}}{A_{1}+a_1}V_{1}\left(1-\dfrac{V_{1}}{K_{V,1}}\right)-d_{V,12}\dfrac{\alpha_{1}^n}{\alpha_{1}^n+c_V^n}\dfrac{1}{1+b_{1}A_{1}}V_{1}+d_{V,21}\dfrac{\alpha_{2}^n}{\alpha_{2}^n+c_V^n}\dfrac{1}{1+b_{2}A_{2}}V_{2}-\mu_{V,1}V_1,\\
     
\dfrac{dA_{1}}{dt}= r_AA_{1}\left(1-\dfrac{A_{1}}{(1-\alpha_{1})K_{A,1}}\right)-d_{A,12}\dfrac{\alpha_{1}^n}{\alpha_{1}^n+c_A^n}A_{1}+d_{A,21}\dfrac{\alpha_{2}^n}{\alpha_{2}^n+c_A^n}A_{2}-\mu_{A,1}A_1,\\
     
     \dfrac{dV_{2}}{dt}= r_V\dfrac{A_{2}}{A_{2}+a_2}V_{2}\left(1-\dfrac{V_{2}}{K_{V,2}}\right)-d_{V,21}\dfrac{\alpha_{2}^n}{\alpha_{2}^n+c_V^n}\dfrac{1}{1+b_{2}A_{2}}V_{2}+d_{V,12}\dfrac{\alpha_{1}^n}{\alpha_{1}^n+c_V^n}\dfrac{1}{1+b_{1}A_{1}}V_{1}-\mu_{V,2}V_2,\\

     \dfrac{dA_{2}}{dt}= r_AA_{2}\left(1-\dfrac{A_{2}}{(1-\alpha_{2})K_{A,2}}\right)-d_{A,21}\dfrac{\alpha_{2}^n}{\alpha_{2}^n+c_A^n}A_{2}+d_{A,12}\dfrac{\alpha_{1}^n}{\alpha_{1}^n+c_A^n}A_{1}-\mu_{A,2}A_2.\\
\end{array}
    \right.
\end{equation}

We set $x = (V_{1}, A_{1},V_{2}, A_{2})$ and
$\mathcal{D}=\mathbb{R}^4_+=\{x\in\mathbb{R}^4:x\geq\textbf{0}\}$. Then model (\ref{2p-model}) can be written in the form 

\begin{equation}\label{ODE-system}
\frac{dx}{dt}=f(x),
\end{equation}
where $f:\mathbb{R}^4\to\mathbb{R}^4$ represents the right hand side of (\ref{2p-model}). Function $f$ is continuous and continuously differentiable on $\mathbb{R}^4$. Thus, according to \cite[Theorem III.10.VI]{Walter1998}, for any initial condition a unique solution exists, at least locally. The vector field defined by $f$ is either tangential or directed inwards on $\partial D$. Therefore, for any initial condition in $\mathcal{D}$ the solution of (\ref{2p-model}) remains in $\mathcal{D}$ for its maximal interval of existence \cite[Theorem III.10.XVI]{Walter1998}.

\section{Results}\label{sec: lin stab analysis}
The phenomenological nature of our model makes possible a stability analysis to determine asymptotic behaviour of model solutions. To do so, it is instructive to start with a simplified model in which no migration between patches occur (\Cref{sec: no migration}), then extend the analysis to a case of one-way migration (\Cref{sec: one way mig}), before finally analysing the full model (\Cref{sec: full model}). For ease of reading, we only state results in these sections, with proofs shown in the Appendix. 

\subsection{No migration}\label{sec: no migration}
We start our analysis by considering the simplified case of no migration ($d_{V,12}=d_{V,21}=d_{A,12}=d_{A,21}=0$) between patches. This could, for example, represent patches which surroundings are sufficiently hostile to prevent any emigration from or migration to the patch. Moreover, these results will form a natural starting point for a numerical bifurcation analysis using continuation in \Cref{sec: numcont}.

     
     


In this case, the dynamics between patches decouple and it is thus sufficient to consider a reduced model for a single patch only. In patch 1, the dynamics are governed by the system of differential equations
\begin{equation}\label{1p-model-no-diffusion}
    \left\{
\begin{array}{l}
     \dfrac{dV_{1}}{dt}= r_V\dfrac{A_{1}}{A_{1}+a_1}V_{1}\left(1-\dfrac{V_{1}}{K_{V,1}}\right)-\mu_{V,1}V_1,\\
     
\dfrac{dA_{1}}{dt}= r_AA_{1}\left(1-\dfrac{A_{1}}{(1-\alpha_{1})K_{A,1}}\right)-\mu_{A,1}A_1.\\
 \end{array}
    \right.
\end{equation}

To develop an understanding of the steady states of \eqref{1p-model-no-diffusion} and their stability, it is intuitive to define
\begin{equation}
 \label{thresholds}
 \mathcal{R}_{0A,i}=\dfrac{r_A}{\mu_{A,i}},\quad \mathcal{R}_{0V,i}=\dfrac{r_V}{\mu_{V,i}}\dfrac{(1-\alpha_{i})K_{A,i}\left(1-\dfrac{1}{\mathcal{R}_{0A,i}}\right)}{(1-\alpha_{i})K_{A,i}\left(1-\dfrac{1}{\mathcal{R}_{0A,i}}\right)+a_i}, \quad i=1,2.
\end{equation}

From a biological point of view, the threshold $\mathcal{R}_{0A,i}$ denotes the animal population `intrinsic' growth rate (ratio of maximum birth rate and death rate). Similarly, the ratio $\dfrac{r_V}{\mu_{V,i}}$ represents the vector population `intrinsic' growth rate. Consequently, $\mathcal{R}_{0V,i}$ quantifies for the animal-mediated, and anthropization-mediated vector population growth rate, on patch $i=1,2$.

By setting to zero the right hand side of \eqref{1p-model-no-diffusion}, and using linear stability analysis, we obtain the following information on the system's steady states. For a visualisation of these conditions, see \Cref{fig: no mig stab}.
\begin{theorem}\label{theorem-global-stability}
\begin{enumerate} 
    \item The extinction state $e_{00}=(0,0)$ is always an equilibrium. It is globally asymptotically stable (GAS) whenever $\mathcal{R}_{0A,1}\leq1$.
    \item The vector extinction state $e_{0A}=\left(0,(1-\alpha_{1})K_{A,1}\left(1-\dfrac{1}{\mathcal{R}_{0A,1}}\right)\right)$ is an equilibrium if $\mathcal{R}_{0A,1}>1$. It is GAS if $\mathcal{R}_{0V,1}\leq1$.  
    \item The coexistence state $$e_{VA}=\left(K_{V,1}\left(1-\dfrac{1}{\mathcal{R}_{0V,1}}\right),(1-\alpha_{1})K_{A,1}\left(1-\dfrac{1}{\mathcal{R}_{0A,1}}\right)\right)$$ is an equilibrium if $\mathcal{R}_{0A,1}>1$ and $\mathcal{R}_{0V,1}>1$. It is GAS in its entire existence region.    
\end{enumerate}
\end{theorem}
\begin{proof}
    See Appendix \ref{AppendixC0}. 
\end{proof}

\begin{figure}
    \centering
    \includegraphics{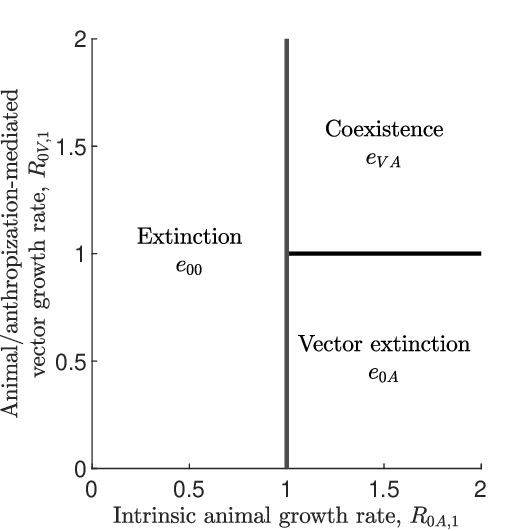}
    \caption{Stable steady states in the decoupled single-patch model (no migration). A classification of stable steady states is shown in the ($\mathcal{R}_{0A,i}$,$\mathcal{R}_{0V,i}$) parameter plane, where $\mathcal{R}_{0A,i}$ denotes the intrinsic animal growth rate and $\mathcal{R}_{0V,i}$ denotes the animal and anthropization-mediated vector growth rate. Note that no bistability occurs.}
    \label{fig: no mig stab}
\end{figure}

\subsection{Biological interpretation of stability results in the absence of migration}
    The linear stability analysis highlights that prevention of extinction requires the animal population's maximum growth rate to be larger than its death rate ($\mathcal{R}_{0A,1}<1$). In particular, since vector growth relies on the presence of animals, no vector persistence is possible, if this condition is violated. For sufficiently high intrinsic animal growth rate ($\mathcal{R}_{0A,1}>1$), long-term dynamics of the system either result in coexistence of vectors and animals, or animal persistence and vector extinction.
    Coexistence requires the vector and anthropization-mediated vector growth rate to be sufficiently large ($\mathcal{R}_{0V,1}>1$). This means that, the half-saturation parameter $a_1$ needs to be sufficiently small, i.e., the efficiency of vectors of finding and feeding on the animal host. That is, acknowledging that vector `intrinsic' offspring number $\dfrac{r_V}{\mu_{V,1}}$ is greater than one, $a_1$ must satisfy 
    
    $$0<a_1<(1-\alpha_{1})K_{A,1}\left(1-\dfrac{1}{\mathcal{R}_{0A,1}}\right)\left(\dfrac{r_V}{\mu_{V,1}}-1\right):=a_1^{\max}.$$
Conversely, if vectors are less efficient at finding or utilizing the host (i.e. $a_1>a_1^{\max}$), the animal-mediated vectors offspring number $\mathcal{R}_{0V,1}$ may fall below one and lead extinction of the vectors. In addition, both $\mathcal{R}_{0V,1}$ and $a_1^{\max}$ are decreasing functions of the anthropization level $\alpha_1$. Hence, an increase of $\alpha_1$ leads a decrease of $\mathcal{R}_{0V,1}$ and $a_1^{\max}$ which, in turn, imply the extinction of vectors in patch 1.

\subsection{One-way migration}\label{sec: one way mig}
We next consider the case in which one of the patches (here chosen to be patch 2) is not anthropized at all, i.e. $\alpha_2$. This leads to a model in which population movement is unidirectional from patch 1 to patch 2. This corresponds, for example, to a system comprising a city (patch 1) and a nearby unanthropized forest. Under this assumption, \eqref{2p-model} becomes

\begin{equation}\label{2p-model-1way-diffusion}
    \left\{
\begin{array}{l}
 \dfrac{dV_{1}}{dt}= r_V\dfrac{A_{1}}{A_{1}+a_1}V_{1}\left(1-\dfrac{V_{1}}{K_{V,1}}\right)-d_{V,12}\dfrac{\alpha_{1}^n}{\alpha_{1}^n+c_V^n}\dfrac{1}{1+b_{1}A_{1}}V_{1}-\mu_{V,1}V_1,\\
     
\dfrac{dA_{1}}{dt}= r_AA_{1}\left(1-\dfrac{A_{1}}{(1-\alpha_{1})K_{A,1}}\right)-d_{A,12}\dfrac{\alpha_{1}^n}{\alpha_{1}^n+c_A^n}A_{1}-\mu_{A,1}A_1,\\
     
      \dfrac{dV_{2}}{dt}= r_V\dfrac{A_{2}}{A_{2}+a_2}V_{2}\left(1-\dfrac{V_{2}}{K_{V,2}}\right)+d_{V,12}\dfrac{\alpha_{1}^n}{\alpha_{1}^n+c_V^n}\dfrac{1}{1+b_{1}A_{1}}V_{1}-\mu_{V,2}V_2,\\

     \dfrac{dA_{2}}{dt}= r_AA_{2}\left(1-\dfrac{A_{2}}{K_{A,2}}\right)+d_{A,12}\dfrac{\alpha_{1}^n}{\alpha_{1}^n+c_A^n}A_{1}-\mu_{A,2}A_2.\\
\end{array}
    \right.
\end{equation}

In \Cref{theoreme-stabilite-1waydiffusion}, we will show the existence of a coexistence equilibrium $E_{V_1A_1V_2A_2} = \left(\bar{V}_1,\bar{A}_1,V_{2,+},A_{2,+}\right)$ with $\bar{V}_1>0$, $\bar{A}_1>0$, $V_{2,+}>0$, $A_{2,+}>0$. 
To investigate the existence and stability of $E_{V_1A_1V_2A_2}$ and other steady states of \eqref{2p-model-1way-diffusion}, we define the following thresholds, noting that some thresholds are functions of the equilibrium densities occuring in the coexistence equilibrium:
\begin{equation}\label{threshold1}
\begin{array}{c}
     d_{V,12}'=d_{V,12}\dfrac{\alpha_{1}^n}{\alpha_{1}^n+c_V^n}, \quad d_{A,12}'=d_{A,12}\dfrac{\alpha_{1}^n}{\alpha_{1}^n+c_A^n}, \quad K_{A,1}'=(1-\alpha_1)K_{A,1},\\     
   
   \mathcal{Q}_{0A,1}=\dfrac{r_A}{\mu_{A,1}+d_{A,12}'},\quad \bar{A}_1=K_{A,1}'\left(1-\dfrac{1}{\mathcal{Q}_{0A,1}}\right),\quad \Delta_A=(r_A-\mu_{A,2})^2+4d_{A,12}'\bar{A}_{1}\dfrac{r_A}{K_{A,2}},\\ 
   
   \bar{V}_1=K_{V,1}\left(1-\dfrac{1}{\mathcal{Q}_{0V,1}}\right),\quad
A_{2,+}=\dfrac{K_{A,2}}{2}\left(1-\dfrac{1}{\mathcal{R}_{0A,2}}+\dfrac{\sqrt{\Delta_A}}{r_A}\right),\\ 

\mathcal{Q}_{0V,1}=\dfrac{r_V\bar{A}_1}{(\bar{A}_1+a_1)\left(\mu_{V,1}+\dfrac{d_{V,12}'}{1+b_1\bar{A}_1}\right)}, \quad
 \mathcal{Q}_{0V,2}=\dfrac{r_V}{\mu_{V,2}}\dfrac{A_{2,+}}{A_{2,+}+a_2}, \quad \bar{V}_2=K_{V,2}\left(1-\dfrac{1}{\mathcal{Q}_{0V,2}}\right),\\
\Delta_V=\left(r_V\dfrac{A_{2,+}}{A_{2,+}+a_2}-\mu_{V,2}\right)^2+4d_{V,12}'\bar{V}_{1}\dfrac{r_V}{K_{V,2}}\dfrac{A_{2,+}}{A_{2,+}+a_2}\dfrac{1}{1+b_1\bar{A}_1},\\
V_{2,+}=\dfrac{K_{V,2}(A_{2,+}+a_2)}{2A_{2,+}}\left(1-\dfrac{\mu_{V,2}}{r_{V}}+\dfrac{\sqrt{\Delta_V}}{r_V}\right). 
\end{array}
\end{equation}



Using these definitions, we can summarise the existence and stability of steady states of \eqref{2p-model-1way-diffusion} as follows. 

\begin{theorem}\label{theoreme-stabilite-1waydiffusion}
    \begin{enumerate}
        \item The extinction equilibrium $E_{0000}=(0,0,0,0)$ is always an equilibrium and is LAS when $\mathcal{Q}_{0A,1}<1$ and $\mathcal{R}_{0A,2}<1$.
    \item The boundary equilibrium $E_{00V_2A_2}=\left(0,0,K_{V,2}\left(1-\dfrac{1}{\mathcal{R}_{0V,2}|_{\alpha_2=0}}\right),K_{A,2}\left(1-\dfrac{1}{\mathcal{R}_{0A,2}}\right)\right)$ exists if $\mathcal{R}_{0A,2}>1$ and $\mathcal{R}_{0V,2}|_{\alpha_2=0}>1$. It is LAS if $\mathcal{Q}_{0A,1}<1$.
    \item The boundary equilibrium $E_{000A_2}=\left(0,0,0,K_{A,2}\left(1-\dfrac{1}{\mathcal{R}_{0A,2}}\right)\right)$ exists if $\mathcal{R}_{0A,2}>1$. It is LAS if $\mathcal{Q}_{0A,1}<1$ and $\mathcal{R}_{0V,2}|_{\alpha_2=0}<1$.
    
    \item The boundary equilibrium $E_{0A_10A_2}=\left(0,\bar{A}_1,0,A_{2,+}\right)$ exists if $\mathcal{Q}_{0A,1}>1$. It is LAS if $\mathcal{Q}_{0V,1}<1$ and $\mathcal{R}_{0V,2}|_{\alpha_2=0}<1$. 
    \item The boundary equilibrium $E_{0A_1V_2A_2}=\left(0,\bar{A}_1,\bar{V}_2 
    ,A_{2,+}\right)$ exists if $\mathcal{Q}_{0A,1}>1$ and $\mathcal{Q}_{0V,2}>1$. It is LAS if $\mathcal{Q}_{0V,1}<1$.
    \item The coexistence equilibrium $
E_{V_1A_1V_2A_2} = \left(\bar{V}_1,\bar{A}_1,V_{2,+},A_{2,+}\right)
$ exists if $\mathcal{Q}_{0A,1}>1$ and $\mathcal{Q}_{0V,1}>1$. It is LAS in its entire existence region.
    \end{enumerate}
\end{theorem}

\begin{proof}
We first establish the existence of equilibria and then deal with their stability results. See Appendix \ref{appendixA} for details.    
\end{proof}

\subsection{Biological interpretation of stability results for the one-way migration}
We first note the relation of these results to results for the single-patch (no migration) model presented in \Cref{theorem-global-stability}. One-way migration from patch 1 to patch 2 enables equilibria in which the population in patch 1 goes extinct, and densities and existence conditions of these equilibria are identical to those derived for the single-patch model with $\alpha_2=0$ (modulo patch label). Stability conditions are also the same, except the addition of $\mathcal{Q}_{0A,1}<1$ as a further constraint. The constant $\mathcal{Q}_{0A,1}$ represents a resilience index of animals of patch 1 experiencing anthropization because it relates the animal's maximum growth rate to the rate animals leave (either through death or migration) the patch. The condition $\mathcal{Q}_{0A,1}<1$ thus means that extinction in patch 1 occurs if the animals' removal rate from the patch is higher than its maximum growth rate. The threshold $\mathcal{Q}_{0A,1}$ is a decreasing function of the anthropization level $\alpha_1$ but more generally, it is a decreasing function of the `displacement' rate out of the patch 1, $d_{A,12}'=d_{A,12}\dfrac{\alpha_{1}^n}{\alpha_{1}^n+c_A^n}$\blank{, which} increases with $\alpha_1$. Therefore, when the parameter $\alpha_1$ increases, animals in patch 1 become less resilient and  their displacement rate out of patch 1 increases. When animals become less present in patch 1, vectors feeding become scarce and then their displacement out of patch 1 will also increase. 

The results also show that prevention of vector extinction in patch 2 requires $\mathcal{R}_{0V,2}|_{\alpha_2=0}>1$. Similar to the interpretation of results for the single patch model this highlights that the half-saturation parameter $a_2$ of the vector reproduction functional response to animal availability \blank{needs to be sufficiently small} to prevent extinction of the vector population in patch 2. That is, acknowledging that vector `intrinsic' offspring number $\dfrac{r_V}{\mu_{V,2}}$ in patch 2 is greater than one, $a_2$ must satisfy 
    
    $$0<a_2<K_{A,2}\left(1-\dfrac{1}{\mathcal{R}_{0A,2}}\right)\left(\dfrac{r_V}{\mu_{V,2}}-1\right):=a_2^{\max}.$$ 
    Recall that, in our setting, patch 2 is free of anthropization, that is $\alpha_2=0.$
If $a_2$ is large (i.e. $a_2>a_2^{\max}$), it indicates that the vectors in patch 2 are less efficient at finding or utilizing their host. This leads to the animal-mediated vector offspring number $\mathcal{R}_{0V,2}|_{\alpha_2=0}$ being larger than one and causes extinction of the vectors in patch 2. 

Finally, we note that $\mathcal{Q}_{0V,1}>1$ is required to prevent vector extinction in patch 1. Thus, $\mathcal{Q}_{0V,1}>1$ that can be understood as the `sustainability' threshold of vectors in patch 1 subject to anthropization. Interestingly, it is a non-monotonic function of the anthropization level $\alpha_1$ of patch 1. This reveals that the anthropization level variations  will non-monotonically impact variations of $\mathcal{Q}_{0V,1}$ and then non-linearly shape the persistence or the extinction of vectors in patch 1. \

\subsection{The general case of two way migration}\label{sec: full model}
We now turn to the general system (\ref{2p-model}) where migration from patch 1 to patch 2 and from patch 2 to patch 1 take place. A linear stability analysis to determine the asymptotic behaviour of the model is possible, but due to the order of the system, the expressions involved are of significant algebraic complexity. This makes interpretation of stability conditions impossible from a biological viewpoint. We therefore refrain from stating full algebraic expressions in the theorem below (but see the appendix for full details), but note that all undefined constants are functions of the entries of the model's Jacobian, evaluated at the respective steady state.

\begin{theorem}\label{theoreme-ekilibre-LAS-2d-migration}

\begin{enumerate}
  %
    \item The extinction equilibrium $E_{0000}=(0,0,0,0)$ exists for all parameter values and is LAS if $r_A<r_{A,\min}$.
    \item There exists at least one vector extinction equilibrium $E_{0A_10A_2}=(0,A_{1,+},0,A_{2,+})$ if $\mathcal{Q}_{0A,1}\geq1$ or $\mathcal{Q}_{0A,2}\geq1$. It is LAS if $B_1<0$ and $B_2>0$. 

    \item There exists at least one positive coexistence equilibrium $E_{V_1A_1V_2A_2}=(V_{1,+},A_{1,+},V_{2,+},A_{2,+})$ if ($\mathcal{Q}_{0A,1}\geq1$ or $\mathcal{Q}_{0A,2}\geq1$) and ($\mathcal{S}_{0A,1}\geq1$ or $\mathcal{S}_{0A,2}\geq1).$ 
It is LAS whenever $C_1>0$ and $C_2>0$.
\end{enumerate}    

\end{theorem}

\begin{proof}
See Appendix \ref{AppendixC}.
\end{proof}

\subsection{Biological interpretation of stability results for two-way migration} \label{sec: numcont}

\Cref{theoreme-ekilibre-LAS-2d-migration} provides exact expressions of steady states and exact conditions on their stability. However, the algebraic complexity of these quantities (see appendix) makes deducing biological interpretations impracticable in most cases. Moreover, it is not guaranteed that the analysis captured all steady states of the system. We therefore employed numerical continuation to construct bifurcation diagrams for the model and provide better insight into how changes in model parameters affect solution dynamics. To do so, we implemented the model system into AUTO-07p \cite{AUTO}, a popular continuation software for ODEs. Given that the main biological questions addressed in this paper revolve around the impact of anthropization on the vector-animal dynamics, we here report results from continuations in which one of the anthropization constants ($0\le \alpha_1,\alpha_2 <1$) was the main bifurcation parameter. This required arbitrary fixing of all other model parameters. For this, we assumed that there is an order of magnitude difference between vector and animal growth and death, vectors have a larger carrying capacity (in the absence of anthropization) than animals, and animals are prompted to migrate at lower anthropization levels than vectors. For all patch-specific parameters we assumed slight differences between patches (e.g., slightly higher mortality in patch 2).  These qualitative assumptions resulted in choices of $ r_V = 3.0 $, $ r_A = 0.3 $, $ a_1 = 0.5 $, $ a_2 = 0.4 $, $ K_{V,1} = 1.0 $, $ K_{V,2} = 1.1 $, $ K_{A,1} = 0.5 $, $ K_{A,2} = 0.4 $, $ d_{V,12} = d_{V,21} = d_{A,12} = d_{A,21} = 0.2 $, $ c_V = 0.7 $, $ c_A = 0.3 $, $ b_1 = 1.0 $, $ b_2 = 1.1 $, $ \mu_{V,1} = 1.0 $, $ \mu_{V,2} = 1.1 $, $ \mu_{A,1} = 0.1 $, $ \mu_{A,2} = 0.11 $. We also performed continuations in other parameters and briefly report on them in \Cref{sec: appendix bif analysis other paras}. 

The construction of the numerical bifurcation diagrams started from the case of no migration (see \Cref{sec: no migration}) for which we were able to determine closed-form expressions of steady states (\Cref{theorem-global-stability}). For each steady state, we initialised the continuation with $d_{V,12}=d_{V,21}=d_{A,12}=d_{A,21}=0$. We then continued the system in the migration parameters to the desired values (either separately or simultaneously, depending on the target values). We subsequently switched the bifurcation parameter to the intended main bifurcation parameter and continued in this parameter. During this continuation, we recorded branching points. We further restarted the continuation at each branching point after branch switching to check whether these branching points lead to steady states not detected in our previous analysis. After performing this procedure for all steady states, we visualised the steady state branches in a bifurcation diagram after deleting any duplicates \blank{arose through} branch switching.

We used the bifurcation diagrams to assess the impact of anthropization on the solution dynamics. We note that due to the symmetry in the system, there is no qualitative difference between the impact of anthropization in patch 1 ($\alpha_1$) and the impact of anthropization in patch 2 ($\alpha_2$). We thus only report outcomes of our continuation in $\alpha_1$. To simplify the analysis, we choose parameters such that $r_A$ is sufficiently large such that changes in $\alpha_1$ do not cause a transition to extinction (\Cref{theoreme-ekilibre-LAS-2d-migration}).

Starting from parameter values that lead to coexistence of all four model densities (see \Cref{fig: contfigure}G for an example solution) in the state of no anthropization in patch 1 ($\alpha_1=0$), increases in $\alpha_1$ can lead to one of three behaviours, depending on the anthropization level of patch 2, $\alpha_2$. For low levels of anthropization in patch 2, increases in anthropization in patch 1 do not lead to a stability change (\Cref{fig: contfigure}A); for higher levels of anthropization in patch 2, increases in anthropization in patch 1 leads to a transition from coexistence of all four model densities to a state in which only the animals persist and the vectors are extinct (\Cref{fig: contfigure}B,D, and \Cref{fig: contfigure}F for an example solution). Note that the transitions between coexistence and vector extinction occur via a single bifurcation. That is, at a critical value of $\alpha_1$, the coexistence state exchanges stability with the vector extinction steady state. Simultaneous vector extinction is a natural consequence of the migration terms in the model, provided migration constants are non-zero; if vectors are present in one of the patches, that patch acts as a source term for the other patch.  There exists also a region of $\alpha_2$ values for which an increase of $\alpha_1$ from zero towards unity leads to a total of three bifurcations: an initial increase of $\alpha_1$ causes a transition from coexistence to vector extinction as described for the previous case. However, a further increase of $\alpha_1$ leads to a bifurcation back to coexistence of all four model densities, before a further increase of $\alpha_1$ causes a final transition to vector extinction (\Cref{fig: contfigure}C). We further observed that for some parameter values, several coexistence and/or several vector extinction equilibria exist. However, we never observed multistability of equilibria. 

The changes to the $\alpha_1$ bifurcation diagrams caused by variations in $\alpha_2$ prompted us to further investigate how bifurcations depend on the two anthropization parameters. For this, we performed a two-parameter continuation of the bifurcation causing an exchange of stability between coexistence and vector extinction equilibria. The resulting stability boundary in the $\alpha_1$-$\alpha_2$ plane is shown in \Cref{fig: contfigure}E and highlights how changes in $\alpha_2$ affect the number of bifurcations observed in a $\alpha_1$ bifurcation diagram. It highlights that changing the value of $\alpha_2$ affects the number of intersections between the bifurcation curve and horizontal transects corresponding to the different cases described above. It shows that the regime leading to multiple stability exchanges in the $\alpha_1$ diagrams is bounded above and below (with respect to $\alpha_2$) by folds in the bifurcation curve (red dots in \Cref{fig: contfigure}). 

We further investigated how changes to other model parameters affect the size of the parameter region (distance between the $\alpha_2$ values at which folds occur) in which multiple bifurcations occur as $\alpha_1$ is changed. To do this, we repeated the two-parameter continuation of the stability boundary between coexistence and vector extinction steady state for other parameters and recorded the fold location. Firstly, this data (visualised in \Cref{fig: fold locs}) revealed that there are parameter values for which no folds occur. Folds only occurred if parameters $r_V$, $r_A$, $a_2$, and $K_{A,2}$ were in respective intervals of finite size, and if parameters $a_1$, $\mu{V,1}$, and $n$ were sufficiently large and parameters $K_{A,1}$$, \mu_{V,2}$, $\mu_{A,1}$ and $\mu_{A,2}$ were sufficiently small. Changes to $K_{V,1}$, $K_{V,2}$, $c_V$, $c_A$, $b_1$, $b_2$, and the migration parameters (both when changing only animal or vector migration parameters and when changing all migration parameters simultaneously) had no impact on the existence of folds in the bifurcation curve.

\begin{figure}
    \centering
    \includegraphics[width=0.85\linewidth]{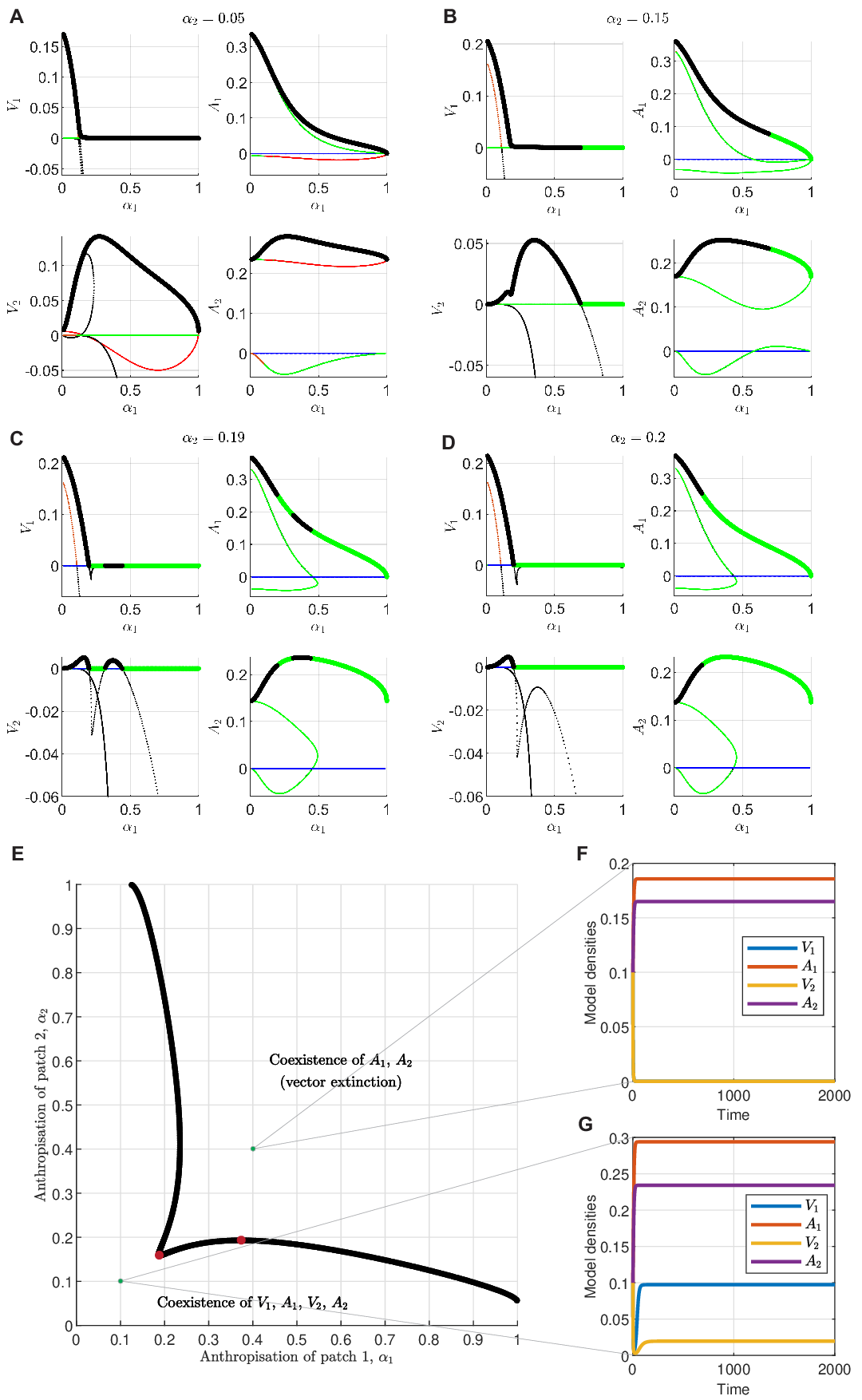}
    \caption{See caption overleaf.}
    \label{fig: contfigure}
\end{figure}
\addtocounter{figure}{-1}
\begin{figure} 
  \caption{\textbf{A-D:} Bifurcation diagrams with respect to $\alpha_1$. The value of $\alpha_2$ varies between plots as indicated by the figure titles. Thick curves represent stable solutions; thin curves represent unstable solutions. The colour of the curve distinguish between different equilibria: black curves are coexistence equilibria; green curves are vector extinction equilibria; blue curves are extinction equilibria; red curves are equilibria in which one of the vector populations goes extinct (never biologically relevant). \textbf{E:} The $\alpha_1$-$\alpha_2$ parameter plane is split according to the stability of equilibria. Red dots indicate the location of folds (with respect to $\alpha_2$ in the bifurcation curve. \textbf{F,G:} Example solution plots corresponding to the parameter values indicated by green dots in \textbf{E}. Note that in \textbf{F}, the solution curves of $V_1$ and $V_2$ overlap.  Across all figures, parameter values are $a_1 = 0.5$, $a_2=0.4$, $b_1=1$, $b_2=1.1$, $c_A=0.3$, $c_V=0.7$, $d_{A,12}=d_{A,21}=d_{V,12}=d_{V,21} = 0.2$, $K_{A,1} = 0.5$, $K_{A,2}=0.4$, $K_{V,1}=1$, $K_{V,2} = 1.1$, $\mu_{A,1} = 0.1$, $\mu_{A,2}=0.11$, $\mu_{V,1}=1$, $\mu_{V,2} = 1.1$, $r_A = 0.3$, $r_V = 3$, $n=2$.}
\end{figure}

\begin{figure}
    \centering
    \includegraphics[width=\linewidth]{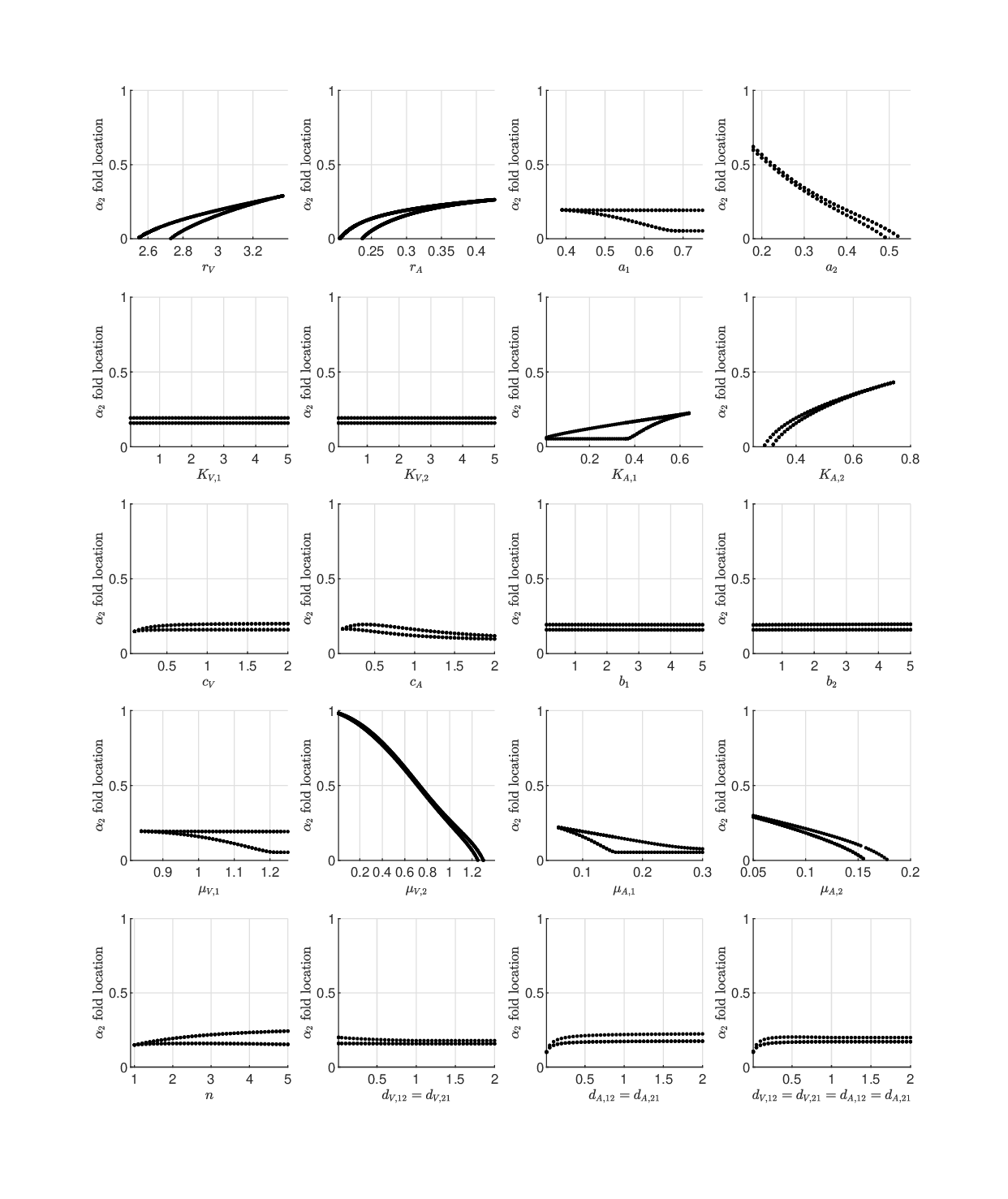}
\caption{Fold location data. Each panel shows how varying a single parameter affects the location of folds (with respect to $\alpha_2$) in the bifurcation curve separating stability regions of the coexistence steady state and the vector extinction steady state.}\label{fig: fold locs}
\end{figure}

\section{Discussion} \label{sec: discussion}



In this paper, we introduced a novel model describing the population dynamics of a disease vector and its animal host in a network of two patches, with migration between both patches being affected by the anthropization of the patches. To the best of our knowledge, this is the first model that accounts for both vector-animal dynamics and anthropization, and thus presents a foundation for revealing how human land use affects population dynamics of disease vectors and their hosts. Through combining a linear stability analysis with numerical continuation, we investigated how increased anthropization associated with increased \blank{land use by humans} \cite{Ellis2013,KathleenLyons2016} affects the population dynamics. We observed complex, i.e., non-monotone dependencies of population levels on anthropization levels, which we discuss in more detail below. These complex relationships observed in our toy model suggest that more information on ecosystems is needed in order to predict how increased land use affects disease vectors, their animal hosts and consequently the eventual disease dynamics (noting that we did not model the latter in this study).

The main result of this paper is the non-monotonic dependency of steady state population densities in one patch on the anthropization level of the connected patch, which can go as far as causing repeated bifurcations back and forth between a vector-extinction and a coexistence equilibrium. That is, an increase in anthropization in patch 1 initially leads to increases in both vector and animal densities in patch 2, before the trend reverses as anthropization in patch 1 becomes larger, eventually resulting in vector extinction on both patches (\Cref{fig: contfigure}). In our model, we assumed that increased anthropization in a patch is detrimental to animals (reduced carrying capacity) and encourages migration away from the patch for both the animals and vectors. This makes the model applicable to species within these constraints, for example the constituents of flea-mammal systems associated with the transmission of human diseases  such as murine typhus \cite{Friggens2010}. Our model is not intended to for species that favour habitats with human activity, such as the dengue, chikungunya, and Zika virus vectors \textit{Aedes aegypti} and \textit{Aedes albopictus} \cite{Kolimenakis2021}. As a result, increases in anthropization in patch 1 leads to increases in the steady state densities in patch 2, provided steady state densities in patch 1 are large. Only increases of anthropization in patch 1 to levels that reduce patch 1 steady state densities to zero or near-zero levels reverse these trends in patch 2 due to decreased net migration into patch 2. 

We further observed that increases in anthropization in patch 1 lead to extinction of vectors in both patches, provided that anthropization in patch 2 is sufficiently large (\Cref{fig: contfigure}). If anthropization in patch 2 is small, then vectors remain at a non-zero steady state density in both patches independent of the anthropization level in patch 1. This is due to our choice of model terms that do not allow extinction of the vector population in a single patch only; for non-zero levels of anthropization, migration between patches always takes place and can thus maintain a vector population for parameter values that would lead to vector extinction in the absence of migration (\Cref{theorem-global-stability}). Related to this, our numerical continuation results provide additional insight. They show that if anthropization in patch 2 is low, there is a significant qualitative shift in steady state densities for increasing anthropization in patch 1, despite the lack of a bifurcation. For low levels of anthropization in patch 1, the vector population in patch 1 is significantly larger than zero, while for larger levels of anthropization in patch 1, the vector population in patch 1 attains near (but non) zero steady state densities. Given the lack of a bifurcation between these two regimes, it is impossible to accurately quantify the threshold causing this transition. However,  this highlights that care is required when interpreting the meaning of the stable coexistence equilibrium. 

For some regions of parameter space, we further identified repeated bifurcations between states of coexistence of vectors and animals and states of vector extinction as the anthropization in patch 1 increases. Our numerical bifurcation analysis of both stable and unstable states provides more insights into these dynamics. Repeated bifurcations can occur, when the gradient of the vector components of the coexistence equilibrium repeatedly changes sign as anthropization in patch 1 increases, leading to the existence of two local maxima on the coexistence branches. While we observed the first local maxima to be always positive (with solutions close to it being stable), the signs of the second local maximum and the local minimum between the local maxima depend on the anthropization level in patch 1 (c.f., \Cref{fig: contfigure} B, C and D). Cases in which both local maxima are positive, but the local minimum is negative leads to repeated bifurcations along the anthropization axis via transcritical bifurcations that occur when the vector densities change sign (\Cref{fig: contfigure} C). This highlights that detailled knowledge of a system is needed to predict how vector populations will respond to changes in anthropization. 

This study presents an initial foray into the question of how anthropization-induced migration affects population dynamics of disease vectors and their hosts. The model used is phenomenological and thus presents many avenues for further study. We restricted our analysis to two patches only. More generally, the framework could be posed on $n \in \N$ patches, with dynamics given by 
\begin{align*}
    \frac{dV_i}{d t} &=  r_V\dfrac{A_{i}}{A_{i}+a_i}V_{i}\left(1-\dfrac{V_{i}}{K_{V,i}}\right)-\dfrac{\alpha_{i}^n}{\alpha_{i}^n+c_V^n}\dfrac{1}{1+b_{i}A_{i}}V_{i} \sum_{k=1}^n d_{V,ik}+ \sum_{k=1}^n d_{V,ki}\dfrac{\alpha_{k}^n}{\alpha_{k}^n+c_V^n}\dfrac{1}{1+b_{k}A_{k}}V_{k} \\ &-\mu_{V,i}V_i, \\   
    \frac{d A_i}{d t} &= r_AA_{i}\left(1-\dfrac{A_{i}}{(1-\alpha_{i})K_{A,i}}\right)-\dfrac{\alpha_{i}^n}{\alpha_{i}^n +c_A^n}A_{i} \sum_{k=1}^n d_{A,ik}+ \sum_{k=1}^n d_{A,ki}\dfrac{\alpha_{k}^n}{\alpha_{k}^n+c_A^n}A_{k}-\mu_{A,i}A_i.
\end{align*}
In this model, connections between the patches would be determined by the migration parameters $d_{V,ik}$ and $d_{A,ik}$ (noting that $d_{V,ii} = d_{V,ii} = 0$) and could provide answers on more complex questions, for example on how the impact of anthropization in one patch propagates through a network of patches depending on the connectivity of the underlying graph. 
Analysis of such a more general model would be more challenging due to the increased number of dimensions. However, numerical bifurcation analysis in particular is expected to be a powerful tool to understand the dynamics in such a system. 

Another line of potential future work is the coupling of the current model to dynamics of a zoonotic vector-borne infectious disease affecting the animal species, humans, or both. It is well known from empirical works that landscape dynamics that influence the movement and distribution of disease vectors and their animal hosts reshape the dynamics of vector-borne diseases \cite{Crowder2013,Ferraguti2021,Fletcher2018,Lambin2010,Shaw2024,Tracey2014}, but modelling studies to further elucidate these dynamics are currently lacking.


\paragraph*{Data availability:}
Computational code to perform the numerical continuation and simulation has been deposited in a Github repository and archived using Zenodo \cite{Coderepo_vector_animal}.

\paragraph*{Acknowledgments:}
IVY-D and PC acknowledge
the support of the AFRICAM project funded by the French Agency for Development (AFD).
IVY-D and LE were supported by a Royal Society International Exchanges Award (IES\textbackslash R1\textbackslash 241236)

\paragraph*{Author contributions (using CRediT):}
O.W. Happi-Tchakount\'e: Methodology, Formal analysis, Writing - Original draft, Writing - Review \& Editing\\ 
I.V. Yatat-Djeumen: Methodology, Formal analysis, Writing - Original draft, Writing - Review \& Editing, Supervision\\
L. Eigentler: Software, Validation, Formal analysis, Investigation, Data curation, Writing - Original draft, Writing - Review \& Editing, Visualization \\
P. Couteron: Methodology, Funding acquisition, Writing - Original draft, Writing - Review \& Editing\\
 
\printbibliography

@article{Tewa2013,
title = {Predator–Prey model with Holling response function of type II and SIS infectious disease},
journal = {Applied Mathematical Modelling},
volume = {37},
number = {7},
pages = {4825-4841},
year = {2013},
issn = {0307-904X},
doi = {https://doi.org/10.1016/j.apm.2012.10.003},
url = {https://www.sciencedirect.com/science/article/pii/S0307904X12006087},
author = {Jean Jules Tewa and Valaire Yatat-Djeumen and Samuel Bowong}
}

@Article{lord_host-Seeking_2017,
	Title = {Host-Seeking Efficiency Can Explain Population Dynamics of the Tsetse Fly {Glossina} Morsitans Morsitans in Response to Host Density Decline},
	Volume = {11},
	Issn = {1935-2735},
	Url = {Https://dx.plos.org/10.1371/journal.pntd.0005730},
	Doi = {10.1371/journal.pntd.0005730},
	Language = {En},
	Number = {7},
	Urldate = {2025-07-01},
	Journal = {Plos Neglected Tropical Diseases},
	Author = {Lord, Jennifer S. and Mthombothi, Zinhle and Lagat, Vitalis K. and Atuhaire, Fatumah and Hargrove, John W.},
	Editor = {Valenzuela, Jesus G.},
	Month = Jul,
	Year = {2017},
	Pages = {E0005730},
}

@mastersthesis{Likeufack2025,
author  = {Likeufack, W. C.},
title   = {modélisation mathématique de la dynamique des mouches tsé-tsé et l’optimisation de la Technique de l’Insecte Stérile (TIS)},
school  = {University of Yaound\'e 1},
year    = {2025},
month   = {Octobre},
type    = {Master's thesis},
}

@article{dumont2025humanwildlifeinteractionstropicalforest,
author = {Dumont, Yves and Hétier, Marc and Yatat-Djeumen, Valaire},
title = {Human–Wildlife Interactions in a Tropical Forest Context: Modeling, Analysis, and Simulations},
journal = {Mathematical Methods in the Applied Sciences},
volume = {n/a},
number = {n/a},
pages = {},
doi = {https://doi.org/10.1002/mma.70777},
url = {https://onlinelibrary.wiley.com/doi/abs/10.1002/mma.70777},
eprint = {https://onlinelibrary.wiley.com/doi/pdf/10.1002/mma.70777}
}

@Misc{Coderepo_vector_animal,
  author    = {Lukas Eigentler},
  date      = {2025},
  title     = {Code repository for O.W. Happi-Tchakounte, I.V. Yatat-Djeumen, L. Eigentler, P. Couteron. Analysis of a two patch model for disease vector-animal dynamics with anthropization-driven migration.},
  doi       = {10.5281/ZENODO.17882418},
  copyright = {GNU General Public License v3.0 only},
  publisher = {Zenodo},
}

@article{dumont_spatio-temporal_2011,
    author = {Dumont, Y. and Dufourd, C.},
    title = {Spatio‐temporal Modeling of Mosquito Distribution},
    journal = {AIP Conference Proceedings},
    volume = {1404},
    number = {1},
    pages = {162-167},
    year = {2011},
    month = {11},
    issn = {0094-243X},
    doi = {10.1063/1.3659916},
    url = {https://doi.org/10.1063/1.3659916},
    eprint = {https://pubs.aip.org/aip/acp/article-pdf/1404/1/162/11645230/162_1_online.pdf},
}

@ARTICLE{Auger2008-lc,
	title     = "The {Ross-Macdonald} model in a patchy environment",
	author    = "Auger, Pierre and Kouokam, Etienne and Sallet, Gauthier and
	Tchuente, Maurice and Tsanou, Berge",
	journal   = "Math. Biosci.",
	publisher = "Elsevier BV",
	volume    =  216,
	number    =  2,
	pages     = "123--131",
	month     =  dec,
	year      =  2008,
	language  = "en"
}

@ARTICLE{Esteva2019-mf,
	title     = "A model for yellow fever with migration",
	author    = "Esteva, Lourdes and Vargas, Cristobal and Yang, Hyun Mo",
	journal   = "Comp and Math Methods",
	publisher = "Hindawi Limited",
	volume    =  1,
	number    =  6,
	month     =  nov,
	year      =  2019,
	copyright = "http://onlinelibrary.wiley.com/termsAndConditions\#vor",
	language  = "en"
}

@ARTICLE{Iggidr2017-jj,
	title     = "Vector borne diseases on an urban environment: The effects of
	heterogeneity and human circulation",
	author    = "Iggidr, A and Koiller, J and Penna, M L F and Sallet, G and
	Silva, M A and Souza, M O",
	journal   = "Ecol. Complex.",
	publisher = "Elsevier BV",
	volume    =  30,
	pages     = "76--90",
	month     =  jun,
	year      =  2017,
	language  = "en"
}

@ARTICLE{Moulay2013-uu,
	title     = "A metapopulation model for chikungunya including populations
	mobility on a large-scale network",
	author    = "Moulay, Djamila and Pign{\'e}, Yoann",
	journal   = "J. Theor. Biol.",
	publisher = "Elsevier BV",
	volume    =  318,
	pages     = "129--139",
	month     =  feb,
	year      =  2013,
	language  = "en"
}

@ARTICLE{Anzo-Hernandez2018-to,
	title     = "The risk matrix of vector-borne diseases in metapopulation
	networks and its relation with local and global {R} 0",
	author    = "Anzo-Hern{\'a}ndez, A and Bonilla-Capilla, B and
	Vel{\'a}zquez-Castro, J and Vel{\'a}zquez-Castro, J and
	Fraguela-Collar, A",
	journal   = "Commun. Nonlinear Sci. Numer. Simul.",
	publisher = "Elsevier BV",
	month     =  jul,
	year      =  2018
}

@ARTICLE{Gimenez-Mujica2025-sf,
	title     = "Estimating the final size of vector-borne epidemics in
	metapopulation networks: methodologies and comparisons",
	author    = "Gim{\'e}nez-Mujica, U J and Vel{\'a}zquez-Castro, J and
	Anzo-Hern{\'a}ndez, A and Herrera-Ram{\'\i}rez, T P and
	Barradas, I",
	journal   = "Bol. Soc. Mat. Mex.",
	publisher = "Springer Science and Business Media LLC",
	volume    =  31,
	number    =  1,
	month     =  mar,
	year      =  2025,
	copyright = "https://www.springernature.com/gp/researchers/text-and-data-mining",
	language  = "en"
}

@ARTICLE{Arino2012-mp,
	title     = "A metapopulation model for malaria with transmission-blocking
	partial immunity in hosts",
	author    = "Arino, Julien and Ducrot, Arnaud and Zongo, Pascal",
	journal   = "J. Math. Biol.",
	publisher = "Springer Science and Business Media LLC",
	volume    =  64,
	number    =  3,
	pages     = "423--448",
	month     =  feb,
	year      =  2012,
	language  = "en"
}

@ARTICLE{Chiaka2021-hw,
  title     = "Land use, environmental, and food consumption patterns in
               sub-Saharan Africa, 2000--2015: A review",
  author    = "Chiaka, Jeffrey Chiwuikem and Zhen, Lin",
  journal   = "Sustainability",
  publisher = "MDPI AG",
  volume    =  13,
  number    =  15,
  pages     = "8200",
  month     =  jul,
  year      =  2021,
  copyright = "https://creativecommons.org/licenses/by/4.0/",
  language  = "en"
}

@ARTICLE{Fetzel2016-xc,
  title     = "Patterns and changes of land use and land-use efficiency in
               Africa 1980--2005: an analysis based on the human appropriation
               of net primary production framework",
  author    = "Fetzel, Tamara and Niedertscheider, Maria and Haberl, Helmut and
               Krausmann, Fridolin and Erb, Karl-Heinz",
  journal   = "Reg. Environ. Change",
  publisher = "Springer Science and Business Media LLC",
  volume    =  16,
  number    =  5,
  pages     = "1507--1520",
  month     =  jun,
  year      =  2016,
  language  = "en"
}

@ARTICLE{Bullock2021-vt,
  title     = "Three decades of land cover change in East Africa",
  author    = "Bullock, Eric L and Healey, Sean P and Yang, Zhiqiang and Oduor,
               Phoebe and Gorelick, Noel and Omondi, Steve and Ouko, Edward and
               Cohen, Warren B",
  journal   = "Land (Basel)",
  publisher = "MDPI AG",
  volume    =  10,
  number    =  2,
  pages     = "150",
  month     =  feb,
  year      =  2021,
  copyright = "https://creativecommons.org/licenses/by/4.0/",
  language  = "en"
}

@ARTICLE{Kimani2021-jb,
  title     = "Impact of human population on land degradation. A critical
               literature review",
  author    = "Kimani, Cassan",
  journal   = "Journal of Environment",
  publisher = "CARI Journals Limited",
  volume    =  1,
  number    =  2,
  pages     = "1--14",
  month     =  jul,
  year      =  2021
}

@article{Friggens2010,
  title = {Anthropogenic Disturbance and the Risk of Flea-Borne Disease Transmission},
  author = {Friggens, Megan M. and Beier, Paul},
  date = {2010-11-01},
  journaltitle = {Oecologia},
  shortjournal = {Oecologia},
  volume = {164},
  number = {3},
  pages = {809--820},
  issn = {1432-1939},
  doi = {10.1007/s00442-010-1747-5},
  url = {https://doi.org/10.1007/s00442-010-1747-5},
  urldate = {2025-12-03},
  langid = {english}
}

@article{Kolimenakis2021,
  title = {The Role of Urbanisation in the Spread of {{Aedes}} Mosquitoes and the Diseases They Transmit—{{A}} Systematic Review},
  author = {Kolimenakis, Antonios and Heinz, Sabine and Wilson, Michael Lowery and Winkler, Volker and Yakob, Laith and Michaelakis, Antonios and Papachristos, Dimitrios and Richardson, Clive and Horstick, Olaf},
  date = {2021-09-09},
  journaltitle = {PLOS Neglected Tropical Diseases},
  shortjournal = {PLOS Neglected Tropical Diseases},
  volume = {15},
  number = {9},
  pages = {e0009631},
  publisher = {Public Library of Science},
  issn = {1935-2735},
  doi = {10.1371/journal.pntd.0009631},
  url = {https://journals.plos.org/plosntds/article?id=10.1371/journal.pntd.0009631},
  urldate = {2025-12-03},
  langid = {english}
}

@TechReport{AUTO,
  author = {Eusebius J. Doedel and Bart E. Oldeman and Alan R. Champneys and Fabio Dercole and Thomas Fairgrieve and Yuri Kuznetsov and Randy Paenroth and Bjorn Sandstede and Xianjun Wang and Chenghai Zhang},
  title  = {AUTO-07p: Continuation and Bifurcation Software for Oridinary Differential Equations},
  year   = {2012},
}

@article{Alkama2016,
  title = {Biophysical Climate Impacts of Recent Changes in Global Forest Cover},
  author = {Alkama, Ramdane and Cescatti, Alessandro},
  date = {2016-02-05},
  journaltitle = {Science},
  volume = {351},
  number = {6273},
  pages = {600--604},
  publisher = {American Association for the Advancement of Science},
  doi = {10.1126/science.aac8083},
  url = {https://www.science.org/doi/10.1126/science.aac8083},
  urldate = {2025-11-20}
}

@article{Cantera2022,
  title = {Low Level of Anthropization Linked to Harsh Vertebrate Biodiversity Declines in {{Amazonia}}},
  author = {Cantera, Isabel and Coutant, Opale and Jézéquel, Céline and Decotte, Jean-Baptiste and Dejean, Tony and Iribar, Amaia and Vigouroux, Régis and Valentini, Alice and Murienne, Jérôme and Brosse, Sébastien},
  date = {2022-06-07},
  journaltitle = {Nature Communications},
  shortjournal = {Nat Commun},
  volume = {13},
  eprint = {35672313},
  eprinttype = {pmid},
  pages = {3290},
  issn = {2041-1723},
  doi = {10.1038/s41467-022-30842-2},
  url = {https://pmc.ncbi.nlm.nih.gov/articles/PMC9174194/},
  urldate = {2025-11-20},
  pmcid = {PMC9174194}
}

@article{DeFrenne2019,
  title = {Global Buffering of Temperatures under Forest Canopies},
  author = {De Frenne, Pieter and Zellweger, Florian and Rodríguez-Sánchez, Francisco and Scheffers, Brett R. and Hylander, Kristoffer and Luoto, Miska and Vellend, Mark and Verheyen, Kris and Lenoir, Jonathan},
  date = {2019-05},
  journaltitle = {Nature Ecology \& Evolution},
  shortjournal = {Nat Ecol Evol},
  volume = {3},
  number = {5},
  pages = {744--749},
  publisher = {Nature Publishing Group},
  issn = {2397-334X},
  doi = {10.1038/s41559-019-0842-1},
  url = {https://www.nature.com/articles/s41559-019-0842-1},
  urldate = {2025-11-20},
  langid = {english}
}

@article{Ellis2013,
  title = {Used Planet: {{A}} Global History},
  shorttitle = {Used Planet},
  author = {Ellis, Erle C. and Kaplan, Jed O. and Fuller, Dorian Q. and Vavrus, Steve and Klein Goldewijk, Kees and Verburg, Peter H.},
  date = {2013-05-14},
  journaltitle = {Proceedings of the National Academy of Sciences},
  volume = {110},
  number = {20},
  pages = {7978--7985},
  publisher = {Proceedings of the National Academy of Sciences},
  doi = {10.1073/pnas.1217241110},
  url = {https://www.pnas.org/doi/10.1073/pnas.1217241110},
  urldate = {2025-11-20}
}

@article{Frishkoff2016a,
  title = {Climate Change and Habitat Conversion Favour the Same Species},
  author = {Frishkoff, Luke O. and Karp, Daniel S. and Flanders, Jon R. and Zook, Jim and Hadly, Elizabeth A. and Daily, Gretchen C. and M'Gonigle, Leithen K.},
  date = {2016},
  journaltitle = {Ecology Letters},
  volume = {19},
  number = {9},
  pages = {1081--1090},
  issn = {1461-0248},
  doi = {10.1111/ele.12645},
  url = {https://onlinelibrary.wiley.com/doi/abs/10.1111/ele.12645},
  urldate = {2025-11-20},
  langid = {english}
}

@article{KathleenLyons2016,
  title = {Holocene Shifts in the Assembly of Plant and Animal Communities Implicate Human Impacts},
  author = {Kathleen Lyons, S. and Amatangelo, Kathryn L. and Behrensmeyer, Anna K. and Bercovici, Antoine and Blois, Jessica L. and Davis, Matt and DiMichele, William A. and Du, Andrew and Eronen, Jussi T. and Tyler Faith, J. and Graves, Gary R. and Jud, Nathan and Labandeira, Conrad and Looy, Cindy V. and McGill, Brian and Miller, Joshua H. and Patterson, David and Pineda-Munoz, Silvia and Potts, Richard and Riddle, Brett and Terry, Rebecca and Tóth, Anikó and Ulrich, Werner and Villaseñor, Amelia and Wing, Scott and Anderson, Heidi and Anderson, John and Waller, Donald and Gotelli, Nicholas J.},
  date = {2016-01},
  journaltitle = {Nature},
  volume = {529},
  number = {7584},
  pages = {80--83},
  publisher = {Nature Publishing Group},
  issn = {1476-4687},
  doi = {10.1038/nature16447},
  url = {https://www.nature.com/articles/nature16447},
  urldate = {2025-11-20},
  langid = {english}
}

@article{Senior2017,
  title = {A Pantropical Analysis of the Impacts of Forest Degradation and Conversion on Local Temperature},
  author = {Senior, Rebecca A. and Hill, Jane K. and González del Pliego, Pamela and Goode, Laurel K. and Edwards, David P.},
  date = {2017},
  journaltitle = {Ecology and Evolution},
  volume = {7},
  number = {19},
  pages = {7897--7908},
  issn = {2045-7758},
  doi = {10.1002/ece3.3262},
  url = {https://onlinelibrary.wiley.com/doi/abs/10.1002/ece3.3262},
  urldate = {2025-11-20},
  langid = {english}
}

@article{Williams2020,
  title = {Human-Dominated Land Uses Favour Species Affiliated with More Extreme Climates, Especially in the Tropics},
  author = {Williams, Jessica J. and Bates, Amanda E. and Newbold, Tim},
  date = {2020},
  journaltitle = {Ecography},
  volume = {43},
  number = {3},
  pages = {391--405},
  issn = {1600-0587},
  doi = {10.1111/ecog.04806},
  url = {https://onlinelibrary.wiley.com/doi/abs/10.1111/ecog.04806},
  urldate = {2025-11-20},
  langid = {english}
}

@article{Hancock2019,
	author  = {Hancock, P. A. and Ritchie, S. A. and Koenraadt, C. J. M. and Scott, T. W. and Hoffmann, A. A. and Godfray, H. C. J.},
	title   = {Predicting the spatial dynamics of Wolbachia infections in Aedes aegypti arbovirus vector populations in heterogeneous landscapes},
	journal = {Journal of Applied Ecology},
	volume  = {56},
	number  = {7},
	pages   = {1674--1686},
	year    = {2019},
	doi     = {10.1111/1365-2664.13423},
}

@article{McCormack2019,
	author  = {McCormack, C. P. and Ghani, A. C. and Ferguson, N. M.},
	title   = {Fine-scale modelling finds that breeding site fragmentation can reduce mosquito population persistence},
	journal = {Communications Biology},
	volume  = {2},
	number  = {1},
	pages   = {273},
	year    = {2019},
	doi     = {10.1038/s42003-019-0525-0},
}

@phdthesis{Lutambi2013,
	author  = {Lutambi, A. M.},
	title   = {Mathematical modelling of Mosquito Dispersal for MalariaVector Control},
	school  = {University of Basel},
	year    = {2013},
	month   = {Jan},
	doi     = {10.5451/UNIBAS-006120662},
}

@ARTICLE{Maneerat2016-th,
	title     = "A spatial agent-based simulation model of the dengue vector
	Aedes aegypti to explore its population dynamics in urban areas",
	author    = "Maneerat, Somsakun and Daud{\'e}, Eric",
	journal   = "Ecol. Modell.",
	publisher = "Elsevier BV",
	volume    =  333,
	pages     = "66--78",
	month     =  aug,
	year      =  2016,
	language  = "en"
}

@article{Roques2016,
	author  = {Roques, L. and Bonnefon, O.},
	title   = {Modelling Population Dynamics in Realistic Landscapes with Linear Elements: A Mechanistic-Statistical Reaction-Diffusion Approach},
	journal = {PLOS ONE},
	volume  = {11},
	number  = {3},
	pages   = {e0151217},
	year    = {2016},
	month   = {Mar},
	doi     = {10.1371/JOURNAL.PONE.0151217},
}

@ARTICLE{Dye2024-sw,
	title     = "Efficacy of Wolbachia-based mosquito control: Predictions of a
	spatially discrete mathematical model",
	author    = "Dye, David and Cain, John W",
	journal   = "PLoS One",
	publisher = "Public Library of Science (PLoS)",
	volume    =  19,
	number    =  3,
	pages     = "e0297964",
	month     =  mar,
	year      =  2024,
	copyright = "http://creativecommons.org/licenses/by/4.0/",
	language  = "en"
}

@ARTICLE{Cui2025-bt,
	title     = "Global dynamics of mosquito population model with seasonality
	and spatial heterogeneity",
	author    = "Cui, Menglei and Li, Fuxiang and Geng, Fengjie",
	journal   = "Math. Methods Appl. Sci.",
	publisher = "Wiley",
	volume    =  48,
	number    =  12,
	pages     = "11850--11862",
	month     =  aug,
	year      =  2025,
	copyright = "http://onlinelibrary.wiley.com/termsAndConditions\#vor",
	language  = "en"
}

@phdthesis{Cailly2011,
	author  = {Cailly, P.},
	title   = {Modélisation de la dynamique spatio-temporelle d’une population de moustiques, sources de nuisances et vecteurs d’agents pathogènes},
	school  = {Universit\'e de Caen Basse-Normandie},
	year    = {2011},
	month   = {Sept},
	language = {French},
}

@article{daSilva2020,
	author  = {da Silva, M. R. and Lugão, P. H. G. and Chapiro, G.},
	title   = {Modeling and simulation of the spatial population dynamics of the Aedes aegypti mosquito with an insecticide application},
	journal = {Parasites \& Vectors},
	volume  = {13},
	number  = {1},
	pages   = {1--13},
	year    = {2020},
	month   = {Nov},
	doi     = {10.1186/S13071-020-04426-2},
}

@phdthesis{nguyen:tel-04687836,
	TITLE = {{Spatial modeling of invasion dynamics : applications to biological control of Aedes spp. (Diptera culicidae)}},
	AUTHOR = {Nguyen, Thi Quynh Nga},
	URL = {https://theses.hal.science/tel-04687836},
	NUMBER = {2024PA131018},
	SCHOOL = {{Universit{\'e} Paris-Nord - Paris XIII}},
	YEAR = {2024},
	MONTH = Jun,
	TYPE = {Theses},
	HAL_ID = {tel-04687836},
	HAL_VERSION = {v1},
}

@ARTICLE{Dufourd2012-dd,
	title     = "Modeling and Simulations of Mosquito Dispersal. The Case of
	Aedes albopictus",
	author    = "Dufourd, Claire and Dumont, Yves",
	journal   = "Biomath",
	publisher = "Institute of Mathematics and Informatics Bulgarian Academy of
	Sciences",
	volume    =  1,
	number    =  2,
	month     =  dec,
    doi = {https://doi.org/10.11145/j.biomath.2012.09.262},
	year      =  2012
}

@ARTICLE{Dufourd2013-uu,
	title     = "Impact of environmental factors on mosquito dispersal in the
	prospect of sterile insect technique control",
	author    = "Dufourd, Claire and Dumont, Yves",
	journal   = "Comput. Math. Appl.",
	publisher = "Elsevier BV",
	volume    =  66,
	number    =  9,
	pages     = "1695--1715",
	month     =  nov,
	year      =  2013,
	copyright = "https://www.elsevier.com/open-access/userlicense/1.0/",
	language  = "en"
}

@article{SAUCEDO2022742,
title = {Host movement, transmission hot spots, and vector-borne disease dynamics on spatial networks},
journal = {Infectious Disease Modelling},
volume = {7},
number = {4},
pages = {742-760},
year = {2022},
issn = {2468-0427},
doi = {https://doi.org/10.1016/j.idm.2022.10.006},
url = {https://www.sciencedirect.com/science/article/pii/S2468042722000835},
author = {Omar Saucedo and Joseph H. Tien}
}

@book{higley1989manual,
  title={Manual of Entomology and Pest Management},
  author={Higley, L.G. and Karr, L.L. and Pedigo, L.P.},
  isbn={9780023933509},
  lccn={89127415},
  url={https://books.google.cm/books?id=2eQnAQAAMAAJ},
  year={1989},
  publisher={Macmillan}
}

@book{tibayrenc2024genetics,
  title={Genetics and evolution of infectious diseases},
  author={Tibayrenc, Michel},
  year={2024},
  publisher={Elsevier}
}

@ARTICLE{Crowder2013,
	title     = "West nile virus prevalence across landscapes is mediated by
	local effects of agriculture on vector and host communities",
	author    = "Crowder, David W and Dykstra, Elizabeth A and Brauner, Jo Marie
	and Duffy, Anne and Reed, Caitlin and Martin, Emily and
	Peterson, Wade and Carri{\`e}re, Yves and Dutilleul, Pierre and
	Owen, Jeb P",
	journal   = "PLoS One",
	publisher = "Public Library of Science (PLoS)",
	volume    =  8,
	number    =  1,
	pages     = "e55006",
	month     =  jan,
	year      =  2013,
	language  = "en"
}

@ARTICLE{Ferraguti2021,
	title     = "Ecological effects on the dynamics of West Nile virus and avian
	Plasmodium: The importance of mosquito communities and landscape",
	author    = "Ferraguti, Martina and Mart{\'\i}nez-de la Puente, Josu{\'e} and
	Figuerola, Jordi",
	journal   = "Viruses",
	publisher = "MDPI AG",
	volume    =  13,
	number    =  7,
	pages     = "1208",
	month     =  jun,
	year      =  2021,
	copyright = "https://creativecommons.org/licenses/by/4.0/",
	language  = "en"
}

@ARTICLE{Fletcher2018,
	title     = "The negative effects of habitat fragmentation operate at the
	scale of dispersal",
	author    = "Fletcher, Jr, Robert J and Reichert, Brian E and Holmes,
	Katherine",
	journal   = "Ecology",
	publisher = "Wiley",
	volume    =  99,
	number    =  10,
	pages     = "2176--2186",
	month     =  oct,
	year      =  2018,
	copyright = "http://onlinelibrary.wiley.com/termsAndConditions\#vor",
	language  = "en"
}

@ARTICLE{Shaw2024,
	title    = "The roles of habitat isolation, landscape connectivity and host
	community in tick-borne pathogen ecology",
	author   = "Shaw, Grace and Lilly, Marie and Mai, Vincent and Clark, Jacoby
	and Summers, Shannon and Slater, Kimetha and Karpathy, Sandor and
	Nakano, Angie and Crews, Arielle and Lawrence, Alexandra and
	Salomon, Jordan and Sambado, Samantha Brianne and Swei, Andrea",
	journal  = "R. Soc. Open Sci.",
	volume   =  11,
	number   =  11,
	pages    = "240837",
	month    =  nov,
	year     =  2024,
	language = "en"
}

@ARTICLE{Tracey2014,
	title     = "An agent‐based movement model to assess the impact of landscape
	fragmentation on disease transmission",
	author    = "Tracey, Jeff A and Bevins, Sarah N and VandeWoude, Sue and
	Crooks, Kevin R",	
	journal   = "Ecosphere",
	publisher = "Wiley",
	volume    =  5,
	number    =  9,
	pages     = "1--24",
	month     =  sep,
	year      =  2014,
	copyright = "http://creativecommons.org/licenses/by/3.0/",
	language  = "en"
}

@ARTICLE{Emmanuel2011,
	title     = "Landscape epidemiology: An emerging perspective in the mapping
	and modelling of disease and disease risk factors",
	author    = "Emmanuel, Nnadi Nnaemeka and Loha, Nimzing and Okolo, Mark
	Ojogba and Ikenna, Onyedibe Kenneth",
	journal   = "Asian Pac. J. Trop. Dis.",
	publisher = "Elsevier BV",
	volume    =  1,
	number    =  3,
	pages     = "247--250",
	month     =  sep,
	year      =  2011,
	language  = "en"
}

@ARTICLE{Patz2004,
	title     = "Unhealthy landscapes: Policy recommendations on land use change
	and infectious disease emergence",
	author    = "Patz, Jonathan A and Daszak, Peter and Tabor, Gary M and
	Aguirre, A Alonso and Pearl, Mary and Epstein, Jon and Wolfe,
	Nathan D and Kilpatrick, A Marm and Foufopoulos, Johannes and
	Molyneux, David and Bradley, David J and {Members of the Working
	Group on Land Use Change Disease Emergence}",
	journal   = "Environ. Health Perspect.",
	publisher = "Environmental Health Perspectives",
	volume    =  112,
	number    =  10,
	pages     = "1092--1098",
	month     =  jul,
	year      =  2004,
	language  = "en"
}

@ARTICLE{Lambin2010,
	title    = "Pathogenic landscapes: Interactions between land, people, disease
	vectors, and their animal hosts",
	author   = "Lambin, Eric F and Tran, Annelise and Vanwambeke, Sophie O and
	Linard, Catherine and Soti, Val{\'e}rie",
	journal  = "International Journal of Health Geographics",
	volume   =  9,
	number   =  1,
	pages    = "54",
	month    =  oct,
	year     =  2010
}

@article{Heath2022,
	title = {Mathematical modelling of the mosquito Aedes polynesiensis in a heterogeneous environment},
	journal = {Mathematical Biosciences},
	volume = {348},
	pages = {108811},
	year = {2022},
	issn = {0025-5564},
	doi = {https://doi.org/10.1016/j.mbs.2022.108811},
	url = {https://www.sciencedirect.com/science/article/pii/S0025556422000244},
	author = {Katherine Heath and Michael B. Bonsall and J\'er\^ome Marie and Herv\'e C. Bossin}
}

@article{Lutambi2013MBS,
	title = {Mathematical modelling of mosquito dispersal in a heterogeneous environment},
	journal = {Mathematical Biosciences},
	volume = {241},
	number = {2},
	pages = {198-216},
	year = {2013},
	issn = {0025-5564},
	doi = {https://doi.org/10.1016/j.mbs.2012.11.013},
	url = {https://www.sciencedirect.com/science/article/pii/S0025556412002337},
	author = {Angelina Mageni Lutambi and Melissa A. Penny and Thomas Smith and Nakul Chitnis}
}

@BOOK{Walter1998,
	title = {Ordinary Differential Equations},
	publisher = {Springer},
	year = {1998},
	author = {Walter, W.}
}

\appendix

\section{Proof of Theorem \ref{theorem-global-stability}}\label{AppendixC0}

Straightforward computations lead that the extinction state $e_{00}=(0,0)$ is always an equilibrium of system \eqref{1p-model-no-diffusion}; when $\mathcal{R}_{0A,1}>1$, then the vector-free state $e_{0A}=\left(0,(1-\alpha_{1})K_{A,1}\left(1-\dfrac{1}{\mathcal{R}_{0A,1}}\right)\right)$ is an equilibrium of system \eqref{1p-model-no-diffusion}. Moreover, assume that $\mathcal{R}_{0A,1}>1$ and $\mathcal{R}_{0V,1}>1$. Then, the coexistence state $$e_{VA}=\left(K_{V,1}\left(1-\dfrac{1}{\mathcal{R}_{0V,1}}\right),(1-\alpha_{1})K_{A,1}\left(1-\dfrac{1}{\mathcal{R}_{0A,1}}\right)\right)$$ is an equilibrium of system \eqref{1p-model-no-diffusion}.  In addition,  in system \eqref{1p-model-no-diffusion}, the dynamics of the state variable $A_{1}$ is uncoupled. Direct computations lead that from \eqref{1p-model-no-diffusion}$_2$, $A_{1}$ converges toward 0 when $\mathcal{R}_{0A,1}\leq1$ while it converges toward $(1-\alpha_{1})K_{A,1}\left(1-\dfrac{1}{\mathcal{R}_{0A,1}}\right)$ when $\mathcal{R}_{0A,1}>1$. The results for the state variable $V_{1}$ are deduced by applying a standard limit system argument, see also \cite{Tewa2013} and references therein.

\section{Proof of Theorem \ref{theoreme-stabilite-1waydiffusion}}\label{appendixA}
    Equilibria of the system \eqref{2p-model-1way-diffusion} are solutions of the system
    
\begin{equation}\label{2p-model-1way-diffusion-equilibria}
    \left\{
\begin{array}{l}
r_V\dfrac{A_{1}}{A_{1}+a_1}V_{1}\left(1-\dfrac{V_{1}}{K_{V,1}}\right)-d_{V,12}\dfrac{\alpha_{1}^n}{\alpha_{1}^n+c_V^n}\dfrac{1}{1+b_{1}A_{1}}V_{1}-\mu_{V,1}V_1=0,\\
     
r_AA_{1}\left(1-\dfrac{A_{1}}{(1-\alpha_{1})K_{A,1}}\right)-d_{A,12}\dfrac{\alpha_{1}^n}{\alpha_{1}^n+c_A^n}A_{1}-\mu_{A,1}A_1=0,\\
     
 r_V\dfrac{A_{2}}{A_{2}+a_2}V_{2}\left(1-\dfrac{V_{2}}{K_{V,2}}\right)+d_{V,12}\dfrac{\alpha_{1}^n}{\alpha_{1}^n+c_V^n}\dfrac{1}{1+b_{1}A_{1}}V_{1}-\mu_{V,2}V_2=0,\\

 r_AA_{2}\left(1-\dfrac{A_{2}}{K_{A,2}}\right)+d_{A,12}\dfrac{\alpha_{1}^n}{\alpha_{1}^n+c_A^n}A_{1}-\mu_{A,2}A_2=0.\\
\end{array}
    \right.
\end{equation}
Recall that
$$
\begin{array}{c}
     d_{V,12}'=d_{V,12}\dfrac{\alpha_{1}^n}{\alpha_{1}^n+c_V^n}, \quad d_{A,12}'=d_{A,12}\dfrac{\alpha_{1}^n}{\alpha_{1}^n+c_A^n}, \quad K_{A,1}'=(1-\alpha_1)K_{A,1}.
\end{array}
$$

From \eqref{2p-model-1way-diffusion-equilibria}, it follows that 

\begin{equation}\label{2p-equilibria}
    \left\{
\begin{array}{l}
V_1=0 \mbox{ or }r_V\dfrac{A_{1}}{A_{1}+a_1}\left(1-\dfrac{V_{1}}{K_{V,1}}\right)-d_{V,12}'\dfrac{1}{1+b_{1}A_{1}}-\mu_{V,1}=0,\\
     
A_1=0 \mbox{ or } r_A\left(1-\dfrac{A_{1}}{K_{A,1}'}\right)-d_{A,12}'-\mu_{A,1}=0,\\
     
 r_V\dfrac{A_{2}}{A_{2}+a_2}V_{2}\left(1-\dfrac{V_{2}}{K_{V,2}}\right)+d_{V,12}'\dfrac{1}{1+b_{1}A_{1}}V_{1}-\mu_{V,2}V_2=0,\\

 r_AA_{2}\left(1-\dfrac{A_{2}}{K_{A,2}}\right)+d_{A,12}'A_{1}-\mu_{A,2}A_2=0.\\
\end{array}
    \right.
\end{equation}

Assume that $A_1=0$ and $V_1=0$. We deduce that
\begin{equation}\label{2p-reduce1}
    \left\{
\begin{array}{l}
V_2=0\mbox{ or } r_V\dfrac{A_{2}}{A_{2}+a_2}\left(1-\dfrac{V_{2}}{K_{V,2}}\right)-\mu_{V,2}=0,\\

A_2=0 \mbox{ or } r_A\left(1-\dfrac{A_{2}}{K_{A,2}}\right)-\mu_{A,2}=0.\\
\end{array}
    \right.
\end{equation}
We therefore obtain:
\begin{itemize}
    \item the extinction equilibrium $E_{0000}=(0,0,0,0)$;
    \item a boundary equilibrium $E_{00V_2A_2}=\left(0,0,K_{V,2}\left(1-\dfrac{1}{\mathcal{R}_{0V,2}}\right),K_{A,2}\left(1-\dfrac{1}{\mathcal{R}_{0A,2}}\right)\right)$ that exists when $\mathcal{R}_{0A,2}>1$ and $\mathcal{R}_{0V,2}>1$;
    \item a boundary equilibrium $E_{000A_2}=\left(0,0,0,K_{A,2}\left(1-\dfrac{1}{\mathcal{R}_{0A,2}}\right)\right)$ if $\mathcal{R}_{0A,2}>1$.
\end{itemize}

Let us set $$\mathcal{Q}_{0A,1}=\dfrac{r_A}{\mu_{A,1}+d_{A,12}'}.$$
Assume that $\mathcal{Q}_{0A,1}>1$, $V_1=0$ and $A_1=\bar{A}_1=K_{A,1}'\left(1-\dfrac{1}{\mathcal{Q}_{0A,1}}\right)$. We obtain from \eqref{2p-equilibria} 
\begin{equation}\label{2p-reduce2}
    \left\{
\begin{array}{l}
V_2=0 \mbox{ or } r_V\dfrac{A_{2}}{A_{2}+a_2}\left(1-\dfrac{V_{2}}{K_{V,2}}\right)-\mu_{V,2}=0,\\

 -\dfrac{r_A}{K_{A,2}} A_{2}^2+\left(r_A-\mu_{A,2}\right)A_2+d_{A,12}'\bar{A}_{1}=0.\\
\end{array}
    \right.
\end{equation}

 Equation \eqref{2p-reduce2}$_2$ is a quadratic equation in $A_2$ that has a unique positive solution. Its discriminant is $$\Delta_A=(r_A-\mu_{A,2})^2+4d_{A,12}'\bar{A}_{1}\dfrac{r_A}{K_{A,2}}>0.$$ The positive solution of \eqref{2p-reduce2}$_2$ is
 $$A_{2,+}=\dfrac{K_{A,2}}{2}\left(1-\dfrac{1}{\mathcal{R}_{0A,2}}+\dfrac{\sqrt{\Delta_A}}{r_A}\right).$$

 We therefore obtain from \eqref{2p-reduce2}$_1$ that $V_2=0$ or $V_2=K_{V,2}\left(1-\dfrac{1}{\mathcal{Q}_{0V,2}}\right)$ where $$
 \mathcal{Q}_{0V,2}=\dfrac{r_V}{\mu_{V,2}}\dfrac{A_{2,+}}{A_{2,+}+a_2}>1.$$ 
We deduce:
\begin{itemize}
    \item a boundary equilibrium $E_{0A_10A_2}=\left(0,K_{A,1}'\left(1-\dfrac{1}{\mathcal{Q}_{0A,1}}\right),0,\dfrac{K_{A,2}}{2}\left(1-\dfrac{1}{\mathcal{R}_{0A,2}}+\dfrac{\sqrt{\Delta_A}}{r_A}\right)\right)$ that exists whenever $\mathcal{Q}_{0A,1}>1$;
    \item a boundary equilibrium $$E_{0A_1V_2A_2}=\left(0,K_{A,1}'\left(1-\dfrac{1}{\mathcal{Q}_{0A,1}}\right),K_{V,2}\left(1-\dfrac{1}{\mathcal{Q}_{0V,2}}\right)   
    ,\dfrac{K_{A,2}}{2}\left(1-\dfrac{1}{\mathcal{R}_{0A,2}}+\dfrac{\sqrt{\Delta_A}}{r_A}\right)\right)$$ that exists whenever $\mathcal{Q}_{0A,1}>1$ and $\mathcal{Q}_{0V,2}>1$.
\end{itemize}
Assume that $A_1=\bar{A}_1=K_{A,1}'\left(1-\dfrac{1}{\mathcal{Q}_{0A,1}}\right)$ and $V_1=\bar{V}_1=K_{V,1}\left(1-\dfrac{1}{\mathcal{Q}_{0V,1}}\right)$ where $\mathcal{Q}_{0A,1}>1$ and $$\mathcal{Q}_{0V,1}=\dfrac{r_V\bar{A}_1}{(\bar{A}_1+a_1)\left(\mu_{V,1}+\dfrac{d_{V,12}'}{1+b_1\bar{A}_1}\right)}>1.$$ We obtain from \eqref{2p-equilibria} 

\begin{equation}\label{2p-reduce4}
    \left\{
\begin{array}{l}
 r_V\dfrac{A_{2}}{A_{2}+a_2}V_{2}\left(1-\dfrac{V_{2}}{K_{V,2}}\right)+d_{V,12}'\dfrac{1}{1+b_{1}\bar{A}_{1}}\bar{V}_{1}-\mu_{V,2}V_2=0,\\

 r_AA_{2}\left(1-\dfrac{A_{2}}{K_{A,2}}\right)+d_{A,12}'\bar{A}_{1}-\mu_{A,2}A_2=0.\\
\end{array}
    \right.
\end{equation}
 Equation \eqref{2p-reduce4}$_2$ is the same as equation \eqref{2p-reduce2}$_2$ and its positive solution is
 $$A_{2,+}=\dfrac{K_{A,2}}{2}\left(1-\dfrac{1}{\mathcal{R}_{0A,2}}+\dfrac{\sqrt{\Delta_A}}{r_A}\right).$$
Similarly, \eqref{2p-reduce4}$_1$ is a quadratic equation in $V_2$ that has a unique positive solution. 

Its discriminant is $$\Delta_V=\left(r_V\dfrac{A_{2,+}}{A_{2,+}+a_2}-\mu_{V,2}\right)^2+4d_{V,12}'\bar{V}_{1}\dfrac{r_V}{K_{V,2}}\dfrac{A_{2,+}}{A_{2,+}+a_2}\dfrac{1}{1+b_1\bar{A}_1}V_1>0.$$
The positive solution of \eqref{2p-reduce4}$_1$ is therefore

$$V_{2,+}=\dfrac{K_{V,2}(A_{2,+}+a_2)}{2A_{2,+}}\left(1-\dfrac{\mu_{V,2}}{r_{V}}+\dfrac{\sqrt{\Delta_V}}{r_V}\right).$$
Thus, we deduce the coexistence equilibrium $E_{V_1A_1V_2A_2}=$
$$\left(K_{V,1}\left(1-\frac{1}{\mathcal{Q}_{0V,1}}\right),K_{A,1}'\left(1-\frac{1}{\mathcal{Q}_{0A,1}}\right),\frac{K_{V,2}(A_{2,+}+a_2)}{2A_{2,+}}\left(1-\frac{\mu_{V,2}}{r_{V}}+\frac{\sqrt{\Delta_V}}{r_V}\right),\frac{K_{A,2}}{2}\left(1-\frac{1}{\mathcal{R}_{0A,2}}+\frac{\sqrt{\Delta_A}}{r_A}\right)\right)$$ that exists whenever $\mathcal{Q}_{0A,1}>1$ and $\mathcal{Q}_{0V,1}>1$. This ends the first part of the proof. In the sequel, we now deal with stability results.

The Jacobian matrix of system \eqref{2p-model-1way-diffusion} at $E_{0000}$ is
$$J_{E_{0000}}=
\begin{pmatrix}
    -d_{V,12}'-\mu_{V,1} & 0 & 0 & 0 \\
    0 & r_A\left(1-\dfrac{1}{\mathcal{Q}_{0A,1}}\right) & 0 & 0\\
    d_{V,12}' & 0 & -\mu_{V,2} & 0\\
    0 & d_{A,12}' & 0 & r_A\left(1-\dfrac{1}{\mathcal{R}_{0A,2}}\right)
\end{pmatrix}.
$$
Therefore, $E_{0000}$ is LAS when $\mathcal{Q}_{0A,1}<1$ and $\mathcal{R}_{0A,2}<1$. It is unstable if either $\mathcal{Q}_{0A,1}>1$ or $\mathcal{R}_{0A,2}>1$.

Assume that equilibrium $E_{00V_2A_2}$ exists, that is $\mathcal{R}_{0A,2}>1$ and $\mathcal{R}_{0V,2}|_{\alpha_2=0}>1$. The Jacobian matrix at $E_{00V_2A_2}$ is
$$J_{E_{00V_2A_2}}=
\begin{pmatrix}
    -d_{V,12}'-\mu_{V,1} & 0 & 0 & 0 \\
    0 & r_A\left(1-\dfrac{1}{\mathcal{Q}_{0A,1}}\right) & 0 & 0\\
    d_{V,12}' & 0 & -r_V\dfrac{A_2}{A_2+a_2}\dfrac{V_2}{K_{V,2}} & r_VV_2\left(1-\dfrac{V_2}{K_{V,2}}\right)\dfrac{a_2}{(A_2+a_2)^2}\\
    0 & d_{A,12}' & 0 & -r_A\dfrac{A_2}{K_{A,2}}
\end{pmatrix}.
$$
Hence, $E_{00V_2A_2}$ is LAS when $\mathcal{Q}_{0A,1}<1$ and unstable if $\mathcal{Q}_{0A,1}>1$.

Assume that equilibrium $E_{000A_2}$ exists, that is $\mathcal{R}_{0A,2}>1$. The Jacobian matrix at $E_{000A_2}$ is
$$J_{E_{000A_2}}=
\begin{pmatrix}
    -d_{V,12}'-\mu_{V,1} & 0 & 0 & 0 \\
    0 & r_A\left(1-\dfrac{1}{\mathcal{Q}_{0A,1}}\right) & 0 & 0\\
    d_{V,12}' & 0 & \mu_{V,2}(\mathcal{R}_{0V,2}|_{\alpha_2=0}-1) & 0\\
    0 & d_{A,12}' & 0 & -r_A\dfrac{A_2}{K_{A,2}}
\end{pmatrix}.
$$
Hence, $E_{000A_2}$ is LAS when $\mathcal{Q}_{0A,1}<1$ and $\mathcal{R}_{0V,2}|_{\alpha_2=0}<1$. It is unstable if either $\mathcal{Q}_{0A,1}>1$ or $\mathcal{R}_{0V,2}|_{\alpha_2=0}>1$.

Assume that equilibrium $E_{0A_10A_2}$ exists, that is $\mathcal{Q}_{0A,1}>1$. The Jacobian matrix at $E_{0A_10A_2}$ is
$$J_{E_{0A_10A_2}}=
\begin{pmatrix}
    \left(d_{V,12}\dfrac{1}{1+b_1A_1}+\mu_{V,1}\right)\left(\mathcal{Q}_{0V,1}-1\right) & 0 & 0 & 0 \\
    0 & -r_A\dfrac{A_1}{K_{A,1}'} & 0 & 0\\
    d_{V,12}'\dfrac{1}{1+b_1A_1} & 0 & \mu_{V,2}(\mathcal{R}_{0V,2}|_{\alpha_2=0}-1) & 0\\
    0 & d_{A,12}' & 0 & -r_A\dfrac{A_2}{K_{A,2}}
\end{pmatrix}.
$$
Hence, $E_{0A_10A_2}$ is LAS when $\mathcal{Q}_{0V,1}<1$ and $\mathcal{R}_{0V,2}|_{\alpha_2=0}<1$. It is unstable if either $\mathcal{Q}_{0V,1}>1$ or $\mathcal{R}_{0V,2}|_{\alpha_2=0}>1$.

Assume that equilibrium $E_{0A_1V_2A_2}$ exists, that is $\mathcal{Q}_{0A,1}>1$ and $\mathcal{Q}_{0V,2}>1$. The Jacobian matrix at $E_{0A_1V_2A_2}$ is $J_{E_{0A_1V_2A_2}}=$
$$
\begin{pmatrix}
    \left(d_{V,12}\frac{1}{1+b_1A_1}+\mu_{V,1}\right)\left(\mathcal{Q}_{0V,1}-1\right) & 0 & 0 & 0 \\
    0 & -r_A\frac{A_1}{K_{A,1}'} & 0 & 0\\
    d_{V,12}'\frac{1}{1+b_1A_1} & 0 & -r_V\frac{A_2}{A_2+a_2}\frac{V_2}{K_{V,2}} & r_VV_2\left(1-\frac{V_2}{K_{V,2}}\right)\frac{a_2}{(A_2+a_2)^2}\\
    0 & d_{A,12}' & 0 & -r_A\frac{A_2}{K_{A,2}}-d_{A,12}'\frac{A_1}{A_2}
\end{pmatrix}.
$$
Hence, $E_{0A_1V_2A_2}$ is LAS when $\mathcal{Q}_{0V,1}<1$ and unstable if $\mathcal{Q}_{0V,1}>1$.

Assume that the coexistence equilibrium $E_{V_1A_1V_2A_2}$ exists, that is $\mathcal{Q}_{0A,1}>1$ and $\mathcal{Q}_{0V,1}>1$. The Jacobian matrix at $E_{V_1A_1V_2A_2}$ is $J_{E_{V_1A_1V_2A_2}}=$
$$
\begin{pmatrix}
    -r_V\frac{A_1}{A_1+a_1}\frac{V_1}{K_{V,1}} & J_{12} & 0 & 0 \\
    0 & -r_A\frac{A_1}{K_{A,1}'} & 0 & 0\\
    d_{V,12}'\frac{1}{1+b_1A_1} &     
    -d_{V,12}'\frac{b_1V_1}{(1+b_1A_1)^2}    
    &     
    -r_V\frac{A_2}{A_2+a_2}\frac{V_2}{K_{V,2}} - d_{V,12}'\frac{1}{1+b_1A_1}\frac{V_1}{V_2}    
    & J_{34}\\
    0 & d_{A,12}' & 0 & -r_A\frac{A_2}{K_{A,2}}-d_{A,12}'\frac{A_1}{A_2}
\end{pmatrix},
$$
where $J_{12}=r_VV_1\left(1-\frac{V_1}{K_{V,1}}\right)\frac{a_1}{(A_1+a_1)^2}+d_{V,12}'\frac{b_1}{(1+b_1A_1)^2}V_1$ and $J_{34}=r_VV_2\left(1-\frac{V_2}{K_{V,2}}\right)\frac{a_2}{(A_2+a_2)^2}.$ Eigenvalues of $J_{E_{V_1A_1V_2A_2}}$ are on its main diagonal and are negative.
Hence, the coexistence equilibrium $E_{V_1A_1V_2A_2}$ is LAS whenever it exits. This ends the proof.

\section{Proof of Theorem \ref{theoreme-ekilibre-LAS-2d-migration}}\label{AppendixC}
The first part of the proof deals with the existence of equilibria of the system (\ref{2p-model}) while the second part will deal with their stability analysis.

 Equilibria of system (\ref{2p-model}) are solutions of the system of equations given by
\begin{equation}\label{2p-model-ekilibre}
    \left\{
\begin{array}{l}
 r_V\dfrac{A_{1}}{A_{1}+a_1}V_{1}\left(1-\dfrac{V_{1}}{K_{V,1}}\right)-d_{V,12}\dfrac{\alpha_{1}^n}{\alpha_{1}^n+c_V^n}\dfrac{1}{1+b_{1}A_{1}}V_{1}+d_{V,21}\dfrac{\alpha_{2}^n}{\alpha_{2}^n+c_V^n}\dfrac{1}{1+b_{2}A_{2}}V_{2}-\mu_{V,1}V_1=0,\\
     
r_AA_{1}\left(1-\dfrac{A_{1}}{(1-\alpha_{1})K_{A,1}}\right)-d_{A,12}\dfrac{\alpha_{1}^n}{\alpha_{1}^n+c_A^n}A_{1}+d_{A,21}\dfrac{\alpha_{2}^n}{\alpha_{2}^n+c_A^n}A_{2}-\mu_{A,1}A_1=0,\\
     
    r_V\dfrac{A_{2}}{A_{2}+a_2}V_{2}\left(1-\dfrac{V_{2}}{K_{V,2}}\right)-d_{V,21}\dfrac{\alpha_{2}^n}{\alpha_{2}^n+c_V^n}\dfrac{1}{1+b_{2}A_{2}}V_{2}+d_{V,12}\dfrac{\alpha_{1}^n}{\alpha_{1}^n+c_V^n}\dfrac{1}{1+b_{1}A_{1}}V_{1}-\mu_{V,2}V_2=0,\\

     r_AA_{2}\left(1-\dfrac{A_{2}}{(1-\alpha_{2})K_{A,2}}\right)-d_{A,21}\dfrac{\alpha_{2}^n}{\alpha_{2}^n+c_A^n}A_{2}+d_{A,12}\dfrac{\alpha_{1}^n}{\alpha_{1}^n+c_A^n}A_{1}-\mu_{A,2}A_2=0.\\
\end{array}
    \right.
\end{equation}
Since $(0,0,0,0)$ is a solution of system \eqref{2p-model-ekilibre}, we deduce that $E_{0000}=(0,0,0,0)$ is the extinction equilibrium of system (\ref{2p-model}). We now assume that, $V_1=0$ and $V_2=0$ that verify \eqref{2p-model-ekilibre}$_1$ and \eqref{2p-model-ekilibre}$_3$. Then, $A_1$ and $A_2$ are positive solutions of the system

\begin{equation}\label{2p-model-ekilibre-reduit}
    \left\{
\begin{array}{l}
r_AA_{1}\left(1-\dfrac{A_{1}}{(1-\alpha_{1})K_{A,1}}\right)-d_{A,12}\dfrac{\alpha_{1}^n}{\alpha_{1}^n+c_A^n}A_{1}+d_{A,21}\dfrac{\alpha_{2}^n}{\alpha_{2}^n+c_A^n}A_{2}-\mu_{A,1}A_1=0,\\
     
 r_AA_{2}\left(1-\dfrac{A_{2}}{(1-\alpha_{2})K_{A,2}}\right)-d_{A,21}\dfrac{\alpha_{2}^n}{\alpha_{2}^n+c_A^n}A_{2}+d_{A,12}\dfrac{\alpha_{1}^n}{\alpha_{1}^n+c_A^n}A_{1}-\mu_{A,2}A_2=0.\\
\end{array}
    \right.
\end{equation}

Recall that
$$
\begin{array}{c}
     d_{V,12}'=d_{V,12}\dfrac{\alpha_{1}^n}{\alpha_{1}^n+c_V^n}, \quad K_{A,1}'=(1-\alpha_1)K_{A,1},\\
   d_{A,12}'=d_{A,12}\dfrac{\alpha_{1}^n}{\alpha_{1}^n+c_A^n}, \quad K_{A,2}'=(1-\alpha_1)K_{A,2}.
\end{array}
$$
From \eqref{2p-model-ekilibre-reduit}$_1$, we obtain that $$A_2=\dfrac{1}{d'_{A,21}}(u_2A_1+u_1)A_1$$ where $u_1=-(r_A-(d'_{A,12}+\mu_{A,1}))$ and
$u_2=\dfrac{r_A}{K'_{A,1}}$. Similarly, from \eqref{2p-model-ekilibre-reduit}$_2$, we have $$A_1=\dfrac{1}{d'_{A,12}}(v_2A_2+v_1)A_2$$ where $v_1=-(r_A-(d'_{A,21}+\mu_{A,2}))$
and $v_2=\dfrac{r_A}{K'_{A,2}}$. It therefore follows that either $A_1=0$ or $A_1$ is a positive solution of the cubic equation
$$h(A_1):=w_3A_1^3+w_2A_1^2+w_1A_1+w_0=0$$ with
$$
\begin{array}{l}
w_3=\dfrac{v_2u_2^2}{d'_{A,21}},\\
w_2=2\dfrac{v_2u_1u_2}{d'_{A,21}},\\
w_1=\dfrac{v_2u_1^2}{d'_{A,21}}+v_1u_2,\\
w_0=u_1v_1-d'_{A,12}d'_{A,21}.
\end{array}
$$
We discuss the number of positive roots of $h$ by using the Descartes's sign rule and it is summarized in Table \ref{table-positive-roots}

\begin{table}[H]\centering
	\begin{tabular}{ccc|c}
		\hline
		$w_2$ & $w_1$  & $w_0$  & Number of positive roots\\
	\hline
	
	+ & + & + & 0 \\
	+ & + & - & 1\\
	+ & - & + & 2 or 0 \\
	- & + & + & 1\\
	- & - & + & 1 \\
	- & + & - & 2 or 0 \\
	+ & - & - & 1 \\
	- & - & - & 0 \\
			
	\hline
	\end{tabular}
	
	\caption{Number of possible positive roots of $h$.}\label{table-positive-roots}
\end{table}

Note that a positive root $A_1^*$ of $h$ will lead to a biologically meaningful equilibrium if $A_1^*>-\dfrac{u_1}{u_2}.$ Direct computations lead that 
$$h\left(-\dfrac{u_1}{u_2}\right)=-d'_{A,12}d'_{A,21}<0 \mbox{ and } \lim\limits_{A_1\rightarrow+\infty}h(A_1)=+\infty.$$ Using the intermediate values theorem, we deduce that $h$ has at least one root $A_1^*\in \left(-\dfrac{u_1}{u_2},+\infty\right)$. Therefore, to ensure that $A_1^*$ is positive, it suffices to ensure that $-\dfrac{u_1}{u_2}>0$ which holds whenever $\Q_{0A,1}>1$. Using again the intermediate values theorem we also deduce that when $\Q_{0A,1}=1$ a positive boundary equilibrium exists. By a symmetry argument, it follows that $\Q_{0A,2}\geq1$ is also a sufficient condition that ensures the existence of a boundary equilibrium.

Let $(A_{1,+},A_{2,+})$ be a positive solution of system \eqref{2p-model-ekilibre-reduit}, then to deal with positive equilibria of system (\ref{2p-model}), we now need to consider the following sub-system

\begin{equation}\label{2p-model-ekilibre-reduit2}
\left\{
\begin{array}{l}
r_V\dfrac{A_{1,+}}{A_{1,+}+a_1}V_{1}\left(1-\dfrac{V_{1}}{K_{V,1}}\right)-d_{V,12}'\dfrac{1}{1+b_{1}A_{1,+}}V_{1}+d_{V,21}'\dfrac{1}{1+b_{2}A_{2,+}}V_{2}-\mu_{V,1}V_1=0,\\

r_V\dfrac{A_{2,+}}{A_{2,+}+a_2}V_{2}\left(1-\dfrac{V_{2}}{K_{V,2}}\right)-d_{V,21}'\dfrac{1}{1+b_{2}A_{2,+}}V_{2}+d_{V,12}'\dfrac{1}{1+b_{1}A_{1,+}}V_{1}-\mu_{V,2}V_2=0.\\

\end{array}
\right.
\end{equation}

Using the same methodology as in system \eqref{2p-model-ekilibre-reduit} and considering the thresholds

\begin{equation}\label{SOV1}    \mathcal{S}_{0V,1}=\dfrac{r_V\dfrac{A_{1,+}}{A_{1,+}+a_1}}{d_{V,12}'\dfrac{1}{1+b_1A_{1,+}}+\mu_{V,1}}, \mbox{ and } \mathcal{S}_{0V,2}=\dfrac{r_V\dfrac{A_{2,+}}{A_{2,+}+a_2}}{d_{V,21}'\dfrac{1}{1+b_2A_{2,+}}+\mu_{V,2}},
\end{equation}
we deduce that:
\begin{itemize}
	\item system \eqref{2p-model-ekilibre-reduit2} may have at most two positive solutions;
	\item if $\mathcal{S}_{0V,1}\geq1$ or $\mathcal{S}_{0V,2}\geq1$, then system \eqref{2p-model-ekilibre-reduit2} has at least a positive solution.
\end{itemize}

We now deal with stability results. In addition to thresholds defined in \eqref{threshold1}, we also set
\begin{equation}\label{SOV01}    \mathcal{Q}_{0A,2}=\dfrac{r_A}{\mu_{A,2}+d_{A,21}'},\quad \mathcal{S}_{0V,1}=\dfrac{r_V\dfrac{A_{1,+}}{A_{1,+}+a_1}}{d_{V,12}'\dfrac{1}{1+b_1A_{1,+}}+\mu_{V,1}}, \mbox{ and } \mathcal{S}_{0V,2}=\dfrac{r_V\dfrac{A_{2,+}}{A_{2,+}+a_2}}{d_{V,21}'\dfrac{1}{1+b_2A_{2,+}}+\mu_{V,2}}.
\end{equation}

$$B_1 = r_V\dfrac{A_{1,+}}{A_{1,+}+a_1}-d'_{V,12}\dfrac{1}{1+b_1A_{1,+}}-\mu_{V,1}+r_V\dfrac{A_{2,+}}{A_{2,+}+a_2}-d'_{V,21}\dfrac{1}{1+b_2A_{2,+}}-\mu_{V,2}$$

$$B_2 = \left(r_V\dfrac{A_{1,+}}{A_{1,+}+a_1}-d'_{V,12}\dfrac{1}{1+b_1A_{1,+}}-\mu_{V,1}\right)\left(r_V\dfrac{A_{2,+}}{A_{2,+}+a_2}-d'_{V,21}\dfrac{1}{1+b_2A_{2,+}}-\mu_{V,2}\right)\\-d'_{V,21}\dfrac{1}{1+b_2A_{2,+}}d'_{V,12}\dfrac{1}{1+b_1A_{1,+}}$$

Let $\Bar{E}=(V_1,A_1,V_2,A_2)$ be an equilibrium of system (\ref{2p-model}). The jacobian matrix of system (\ref{2p-model}) at $\Bar{E}$ assumes the form

\begin{equation}
    \label{jacobian-matrice-general22}
    J(\Bar{E})=\begin{pmatrix}
        J_{11}(\Bar{E}) & J_{12}(\Bar{E}) & J_{13}(\Bar{E}) & J_{14}(\Bar{E})\\
        0 & J_{22}(\Bar{E}) & 0 & J_{24}(\Bar{E})\\
        J_{31}(\Bar{E}) & J_{32}(\Bar{E}) & J_{33}(\Bar{E}) & J_{34}(\Bar{E})\\
        0 & J_{42}(\Bar{E}) & 0 & J_{44}(\Bar{E})\\
    \end{pmatrix}
\end{equation}
where
$$
\begin{array}{l}
  J_{11}(\Bar{E})= r_V\dfrac{A_1}{A_1+a_1}\left(1-\dfrac{2V_1}{K_{V,1}}\right)-d'_{V,12}\dfrac{1}{1+b_1A_1}-\mu_{V,1},     \\
   J_{12}(\Bar{E})=   r_V\dfrac{a_1}{(A_1+a_1)^2}V_1\left(1-\dfrac{V_1}{K_{V,1}}\right)+d'_{V,12}\dfrac{b_1}{(1+b_1A_1)^2}V_1,\\
   J_{13}(\Bar{E})=d'_{V,21}\dfrac{1}{1+b_2A_2},\\
   J_{14}(\Bar{E})=-d'_{V,21}\dfrac{b_2}{(1+b_2A_2)^2}V_2,\\
    J_{22}(\Bar{E})= r_A\left(1-\dfrac{2A_1}{K'_{A,1}}\right)-d'_{A,12}-\mu_{A,1},\\
    J_{24}(\Bar{E})= d'_{A,21},\\
    J_{31}(\Bar{E})=d'_{V,12}\dfrac{1}{1+b_1A_1},\\
    J_{32}(\Bar{E})=-d'_{V,12}\dfrac{b_1}{(1+b_1A_1)^2}V_1,\\
   J_{33}(\Bar{E})=r_V\dfrac{A_2}{A_2+a_2}\left(1-\dfrac{2V_2}{K_{V,2}}\right)-d'_{V,21}\dfrac{1}{1+b_2A_2}-\mu_{V,2}, \\
   J_{34}(\Bar{E})=   r_V\dfrac{a_2}{(A_2+a_2)^2}V_2\left(1-\dfrac{V_2}{K_{V,2}}\right)+d'_{V,21}\dfrac{b_2}{(1+b_2A_2)^2}V_2,\\
   J_{42}(\Bar{E})= d'_{A,12},\\
   J_{44}(\Bar{E})= r_A\left(1-\dfrac{2A_2}{K'_{A,2}}\right)-d'_{A,21}-\mu_{A,2}.\\
\end{array}
$$
$\mathbf{i)}$ Assuming that the equilibrium $\Bar{E}$ is the extinction equilibrium, that is $\Bar{E}=E_{0000}=(0,0,0,0)$. Then, following (\ref{jacobian-matrice-general22}), the jacobian matrix of system (\ref{2p-model}) at $E_{0000}$ is 
\begin{equation*}
    \label{jacobian-matrice-general}
    J(E_{0000})=\begin{pmatrix}
        -d'_{V,12}-\mu_{V,1} & 0 & d'_{V,21} & 0\\
        0 & r_A-d'_{A,12}-\mu_{A,1} & 0 & d'_{A,21}\\
        d'_{V,12} & 0 & -d'_{V,21}-\mu_{V,2} & 0\\
        0 & d'_{A,12} & 0 & r_A-d'_{A,21}-\mu_{A,2}\\
    \end{pmatrix}.
\end{equation*}
The characteristic polynomial of $J(E_{0000})$ assumes the form

$$p(\lambda)=p_1(\lambda)p_2(\lambda),$$
where
$$p_1(\lambda)=\lambda^2+\lambda(d'_{V,12}+\mu_{V,1}+d'_{V,21}+\mu_{V,2})+d'_{V,12}\mu_{V,2}+\mu_{V,1}(d'_{V,21}+\mu_{V,2})$$
and
$$p_2(\lambda)=\lambda^2-\lambda(2r_A-d'_{A,12}-\mu_{A,1}-d'_{A,21}-\mu_{A,2})+ (r_A-\mu_{A,1})(r_A-d'_{A,21}-\mu_{A,2})-d'_{A,12}(r_A-\mu_{A,2}).$$
Since $d'_{V,12}+\mu_{V,1}+d'_{V,21}+\mu_{V,2}>0$ and $d'_{V,12}\mu_{V,2}+\mu_{V,1}(d'_{V,21}+\mu_{V,2})>0$, we deduce that the extinction equilibrium is locally asymptotically stable (LAS) whenever roots of $p_2$ have negative real parts. That is, if $$2r_A-d'_{A,12}-\mu_{A,1}-d'_{A,21}-\mu_{A,2}<0$$ and $$(r_A-\mu_{A,1})(r_A-d'_{A,21}-\mu_{A,2})-d'_{A,12}(r_A-\mu_{A,2})>0.$$
We have that 
\begin{equation}\label{trace22}
    2r_A-d'_{A,12}-\mu_{A,1}-d'_{A,21}-\mu_{A,2}<0  \Leftrightarrow r_A<\dfrac{d'_{A,12}+\mu_{A,1}+d'_{A,21}+\mu_{A,2}}{2}:=\bar{r}_A.
\end{equation}

Moreover, 

$$
\begin{array}{l}
(r_A-\mu_{A,1})(r_A-d'_{A,21}-\mu_{A,2})-d'_{A,12}(r_A-\mu_{A,2})\\
=
r_A^2-r_A(d'_{A,21}+\mu_{A,2}+\mu_{A,1}+d'_{A,12})+\mu_{A,1}(d'_{A,21}+\mu_{A,2})+d'_{A,12}\mu_{A,2}\\
     :=q(r_A). 
\end{array}
$$
We have 
$$
\begin{array}{lcl}
 \Delta_{r_A}&=&(d'_{A,21}+\mu_{A,2}+\mu_{A,1}+d'_{A,12})^2-4(\mu_{A,1}(d'_{A,21}+\mu_{A,2})+d'_{A,12}\mu_{A,2}), \\
     &=&\left(d'_{A,21} -d'_{A,12} + \mu_{A,2} -\mu_{A,1}\right)^2 + 4 d'_{A,21} d'_{A,12},\\
     &\ge& 0. 
\end{array}
$$

$$\Delta_{r_A}=
\left(d'_{A,21} -d'_{A,12} + \mu_{A,2} -\mu_{A,1}\right)^2 + 4 d'_{A,21} d'_{A,12},$$
%
$$\bar{r}_A=\dfrac{d'_{A,12}+\mu_{A,1}+d'_{A,21}+\mu_{A,2}}{2},$$
$$r_{A,\min}=\bar{r}_A-\dfrac{\sqrt{\Delta_{r_A}}}{2}\quad \mbox{ and } \quad r_{A,\max}=\bar{r}_A+\dfrac{\sqrt{\Delta_{r_A}}}{2}.$$

We also set 
 $$   \begin{array}{l}
J_{1,1}=-r_V\dfrac{A_1}{A_1+a_1}\dfrac{V_1}{K_{V,1}}-d'_{V,21}\dfrac{1}{1+b_2A_2}\dfrac{V_2}{V_1},     \\
J_{1,2}=   r_V\dfrac{a_1}{(A_1+a_1)^2}V_1\left(1-\dfrac{V_1}{K_{V,1}}\right)+d'_{V,12}\dfrac{b_1}{(1+b_1A_1)^2}V_1,\\
J_{1,3}=d'_{V,21}\dfrac{1}{1+b_2A_2},\\
J_{1,4}=-d'_{V,21}\dfrac{b_2}{(1+b_2A_2)^2}V_2,\\
J_{2,2}=-r_A\dfrac{A_1}{K'_{A,1}}-d'_{A,21}\dfrac{A_2}{A_1},\\
J_{2,4}= d'_{A,21},\\
J_{3,1}=d'_{V,12}\dfrac{1}{1+b_1A_1},\\
J_{3,2}=-d'_{V,12}\dfrac{b_1}{(1+b_1A_1)^2}V_1,\\
J_{3,3}=-r_V\dfrac{A_2}{A_2+a_2}\dfrac{V_2}{K_{V,2}}-d'_{V,12}\dfrac{1}{1+b_1A_1}\dfrac{V_1}{V_2}, \\
J_{3,4}=   r_V\dfrac{a_2}{(A_2+a_2)^2}V_2\left(1-\dfrac{V_2}{K_{V,2}}\right)+d'_{V,21}\dfrac{b_2}{(1+b_2A_2)^2}V_2,\\
J_{4,2}= d'_{A,12},\\
J_{4,4}= -r_A\dfrac{A_2}{K'_{A,2}}-d'_{A,12}\dfrac{A_1}{A_2},\\
\end{array}
$$
where for simplicity, $(V_{1},A_{1},V_{2},A_{2})=(V_{1,+},A_{1,+},V_{2,+},A_{2,+}).$ We also consider the two conditions

\begin{dmath*}
  C_1=  -J_{{4,2}}J_{{2,4}}+J_{{4,4}}J_{{3,3}}+J_{{4,4}}J_{{2,2}}+J_{{4,4}}J_{
{1,1}}-J_{{3,1}}J_{{1,3}}+J_{{3,3}}J_{{2,2}}+J_{{3,3}}J_{{1,1}}+J_{{2,
2}}J_{{1,1}} 
      -{\frac {N_1}{-J_{{4,4}}-J_{{3,3}}-J_{{2,2}}-J_{
{1,1}}}}
\end{dmath*}
and
\begin{dmath*}
  C_2=   
J_{{4,2}}J_{{2,4}}J_{{3,3}}+J_{{4,2}}J_{{2,4}}J_{{1,1}}+J_{{4,4}}J_{{3
,1}}J_{{1,3}}-J_{{4,4}}J_{{3,3}}J_{{2,2}}-J_{{4,4}}J_{{3,3}}J_{{1,1}}-
J_{{4,4}}J_{{2,2}}J_{{1,1}}+J_{{3,1}}J_{{1,3}}J_{{2,2}}-J_{{3,3}}J_{{2
,2}}J_{{1,1}}
- \left( -J_{{4,4}}-J_{{3,3}}-J_{{2,2}}-J_{{1,1}}
 \right)  \left( J_{{4,2}}J_{{2,4}}J_{{3,1}}J_{{1,3}}-J_{{4,2}}J_{{2,4
}}J_{{3,3}}J_{{1,1}}-J_{{4,4}}J_{{3,1}}J_{{1,3}}J_{{2,2}}+J_{{4,4}}J_{
{3,3}}J_{{2,2}}J_{{1,1}} \right)  \left( -J_{{4,2}}J_{{2,4}}+J_{{4,4}}
J_{{3,3}}+J_{{4,4}}J_{{2,2}}+J_{{4,4}}J_{{1,1}}-J_{{3,1}}J_{{1,3}}+J_{
{3,3}}J_{{2,2}}+J_{{3,3}}J_{{1,1}}+J_{{2,2}}J_{{1,1}}-{\frac {N_2}{-J_{{4,4}}-J_{{3,3}}-J_{{2,2}}-J_{{1,1}}}} \right) ^{-1}
\end{dmath*}
where

$$
\begin{array}{lcl}
  N_1&=&J_{{4,2}}J_{{2,4}}J_{{3,3}}+J_{{4,2}}J_{{2,4}}J_{
{1,1}}+J_{{4,4}}J_{{3,1}}J_{{1,3}}-J_{{4,4}}J_{{3,3}}J_{{2,2}}-J_{{4,4
}}J_{{3,3}}J_{{1,1}}\\
&&-J_{{4,4}}J_{{2,2}}J_{{1,1}}+J_{{3,1}}J_{{1,3}}J_{
{2,2}}-J_{{3,3}}J_{{2,2}}J_{{1,1}}.   \\
\end{array}
$$
and
$$
\begin{array}{lcl}
N_2 &=& J_{{4,2}
}J_{{2,4}}J_{{3,3}}+J_{{4,2}}J_{{2,4}}J_{{1,1}}+J_{{4,4}}J_{{3,1}}J_{{
1,3}}-J_{{4,4}}J_{{3,3}}J_{{2,2}}\\
&&-J_{{4,4}}J_{{3,3}}J_{{1,1}}-J_{{4,4}
}J_{{2,2}}J_{{1,1}}+J_{{3,1}}J_{{1,3}}J_{{2,2}}-J_{{3,3}}J_{{2,2}}J_{{
1,1}}.\\  \\

\end{array}
$$

 Then $q$ has two positive roots
    $$r_{A,\min}=\bar{r}_A-\dfrac{\sqrt{\Delta_{r_A}}}{2}\quad \mbox{ and } \quad r_{A,\max}=\bar{r}_A+\dfrac{\sqrt{\Delta_{r_A}}}{2}.$$ 
    Moreover, $q(r_A)>0$ for all $r_A\in(0,r_{A,\min})$ or $r_A>r_{A,\max}$. and $q(r_A)<0$ for all $r_A\in(r_{A,\min},r_{A,\max})$. Therefore, taking into account \eqref{trace22}, the extinction equilibrium is LAS whenever $r_A<r_{A,\min}$.

$\mathbf{ii)}$ Assume that a boundary equilibrium $E_{0A_10A_2}=(0,A_{1,+},0,A_{2,+})$ exists. That is, assume that $\mathcal{Q}_{0A,1}\geq1$ or $\mathcal{Q}_{0A,2}\geq1$. Following (\ref{jacobian-matrice-general22}), the jacobian matrix of system (\ref{2p-model}) at $E_{0A_10A_2}$ is 
\begin{equation*}
     J(E_{0A_10A_2})=\begin{pmatrix}
        j_{11} & 0 & d'_{V,21}\dfrac{1}{1+b_2A_{2,+}} & 0\\
        0 & j_{22} & 0 & d'_{A,21}\\
        d'_{V,12}\dfrac{1}{1+b_1A_{1,+}} & 0 & j_{33} & 0\\
        0 & d'_{A,12} & 0 & j_{44}\\
    \end{pmatrix}
\end{equation*}
where
$$
\begin{array}{l}
j_{11}= r_V\dfrac{A_{1,+}}{A_{1,+}+a_1}-d'_{V,12}\dfrac{1}{1+b_1A_{1,+}}-\mu_{V,1},     \\
j_{22}= -r_A\dfrac{A_{1,+}}{K'_{A,1}}-d'_{A,21}\dfrac{A_{2,+}}{A_{1,+}},\\
j_{33}=r_V\dfrac{A_{2,+}}{A_{2,+}+a_2}-d'_{V,21}\dfrac{1}{1+b_2A_{2,+}}-\mu_{V,2}, \\
j_{44}= -r_A\dfrac{A_{2,+}}{K'_{A,2}}-d'_{A,12}\dfrac{A_{1,+}}{A_{2,+}}.\\
\end{array}
$$

The characteristic polynomial of $J(E_{0A_10A_2})$ assumes the form

$$p(\lambda)=p_3(\lambda)p_4(\lambda),$$
where
$$p_3(\lambda)=\lambda^2-\lambda(j_{11}+j_{33})+j_{11}j_{33}-d'_{V,21}\dfrac{1}{1+b_2A_{2,+}}d'_{V,12}\dfrac{1}{1+b_1A_{1,+}}$$
and
$$p_4(\lambda)=\lambda^2-\lambda(j_{22}+j_{44})+j_{22}j_{44}-d'_{A,12}d'_{A,21}.$$
Since $j_{22}+j_{44}<0$ and $j_{22}j_{44}-d'_{A,12}d'_{A,21}=r_A\dfrac{A_{1,+}}{K'_{A,1}}\left(r_A\dfrac{A_{2,+}}{K'_{A,2}}+d'_{A,12}\dfrac{A_{1,+}}{A_{2,+}}\right)+r_Ad_{A,21}'\dfrac{A_{2,+}^2}{A_{1,+}K'_{A,2}}>0$, we deduce that the boundary equilibrium $E_{0A_10A_2}$ is LAS whenever roots of $p_3$ have negative real parts. That is, if $$B_1:=j_{11}+j_{33}<0$$ and $$B_2:=j_{11}j_{33}-d'_{V,21}\dfrac{1}{1+b_2A_{2,+}}d'_{V,12}\dfrac{1}{1+b_1A_{1,+}}>0.$$

$\mathbf{iii)}$ Assume that a positive coexistence equilibrium $E_{V_1A_1V_2A_2}=(V_{1,+},A_{1,+},V_{2,+},A_{2,+})$ exists. That is, assume that ($\mathcal{Q}_{0A,1}\geq1$ or $\mathcal{Q}_{0A,2}\geq1$) and ($\mathcal{S}_{0A,1}\geq1$ or $\mathcal{S}_{0A,2}\geq1$). Following (\ref{jacobian-matrice-general22}), the jacobian matrix of system (\ref{2p-model}) at $E_{V_1A_1V_2A_2}=\bar{E}=(V_1,A_1,V_2,A_2)$ is 
\begin{equation*}
\label{jacobian-matrice-general}
J(\bar{E})=\begin{pmatrix}
J_{1,1} & J_{1,2} & J_{1,3} & J_{1,4}\\
0 & J_{2,2} & 0 & J_{2,4}\\
J_{3,1} & J_{3,2} & J_{3,3} & J_{3,4}\\
0 & J_{4,2} & 0 & J_{4,4}\\
\end{pmatrix}
\end{equation*}
where
$$   \begin{array}{l}
J_{1,1}=-r_V\dfrac{A_1}{A_1+a_1}\dfrac{V_1}{K_{V,1}}-d'_{V,21}\dfrac{1}{1+b_2A_2}\dfrac{V_2}{V_1},     \\
J_{1,2}=   r_V\dfrac{a_1}{(A_1+a_1)^2}V_1\left(1-\dfrac{V_1}{K_{V,1}}\right)+d'_{V,12}\dfrac{b_1}{(1+b_1A_1)^2}V_1,\\
J_{1,3}=d'_{V,21}\dfrac{1}{1+b_2A_2},\\
J_{1,4}=-d'_{V,21}\dfrac{b_2}{(1+b_2A_2)^2}V_2,\\
J_{2,2}=-r_A\dfrac{A_1}{K'_{A,1}}-d'_{A,21}\dfrac{A_2}{A_1},\\
J_{2,4}= d'_{A,21},\\
J_{3,1}=d'_{V,12}\dfrac{1}{1+b_1A_1},\\
J_{3,2}=-d'_{V,12}\dfrac{b_1}{(1+b_1A_1)^2}V_1,\\
J_{3,3}=-r_V\dfrac{A_2}{A_2+a_2}\dfrac{V_2}{K_{V,2}}-d'_{V,12}\dfrac{1}{1+b_1A_1}\dfrac{V_1}{V_2}, \\
J_{3,4}=   r_V\dfrac{a_2}{(A_2+a_2)^2}V_2\left(1-\dfrac{V_2}{K_{V,2}}\right)+d'_{V,21}\dfrac{b_2}{(1+b_2A_2)^2}V_2,\\
J_{4,2}= d'_{A,12},\\
J_{4,4}= -r_A\dfrac{A_2}{K'_{A,2}}-d'_{A,12}\dfrac{A_1}{A_2}.\\
\end{array}
$$
The coexistence equilibrium is LAS if the following Routh-Hurwitz conditions are satisfied:
$$C_1>0 \mbox{ and } C_2>0 \mbox{ and } C_3>0 \mbox{ and } C_4>0$$
where
\begin{dmath*}
  C_1=  -J_{{4,2}}J_{{2,4}}+J_{{4,4}}J_{{3,3}}+J_{{4,4}}J_{{2,2}}+J_{{4,4}}J_{
{1,1}}-J_{{3,1}}J_{{1,3}}+J_{{3,3}}J_{{2,2}}+J_{{3,3}}J_{{1,1}}+J_{{2,
2}}J_{{1,1}} 
      -{\frac {J_{{4,2}}J_{{2,4}}J_{{3,3}}+J_{{4,2}}J_{{2,4}}J_{
{1,1}}+J_{{4,4}}J_{{3,1}}J_{{1,3}}-J_{{4,4}}J_{{3,3}}J_{{2,2}}-J_{{4,4
}}J_{{3,3}}J_{{1,1}}}{-J_{{4,4}}-J_{{3,3}}-J_{{2,2}}-J_{
{1,1}}}}
-{\frac {-J_{{4,4}}J_{{2,2}}J_{{1,1}}+J_{{3,1}}J_{{1,3}}J_{
{2,2}}-J_{{3,3}}J_{{2,2}}J_{{1,1}}}{-J_{{4,4}}-J_{{3,3}}-J_{{2,2}}-J_{
{1,1}}}}
\end{dmath*}

\begin{dmath*}
  C_2=   
J_{{4,2}}J_{{2,4}}J_{{3,3}}+J_{{4,2}}J_{{2,4}}J_{{1,1}}+J_{{4,4}}J_{{3
,1}}J_{{1,3}}-J_{{4,4}}J_{{3,3}}J_{{2,2}}-J_{{4,4}}J_{{3,3}}J_{{1,1}}-
J_{{4,4}}J_{{2,2}}J_{{1,1}}+J_{{3,1}}J_{{1,3}}J_{{2,2}}-J_{{3,3}}J_{{2
,2}}J_{{1,1}}
- \left( -J_{{4,4}}-J_{{3,3}}-J_{{2,2}}-J_{{1,1}}
 \right)  \left( J_{{4,2}}J_{{2,4}}J_{{3,1}}J_{{1,3}}-J_{{4,2}}J_{{2,4
}}J_{{3,3}}J_{{1,1}}-J_{{4,4}}J_{{3,1}}J_{{1,3}}J_{{2,2}}+J_{{4,4}}J_{
{3,3}}J_{{2,2}}J_{{1,1}} \right)  \left( -J_{{4,2}}J_{{2,4}}+J_{{4,4}}
J_{{3,3}}+J_{{4,4}}J_{{2,2}}+J_{{4,4}}J_{{1,1}}-J_{{3,1}}J_{{1,3}}+J_{
{3,3}}J_{{2,2}}+J_{{3,3}}J_{{1,1}}+J_{{2,2}}J_{{1,1}}-{\frac {N_2}{-J_{{4,4}}-J_{{3,3}}-J_{{2,2}}-J_{{1,1}}}} \right) ^{-1},
\end{dmath*}

$N_2=J_{{4,2}
}J_{{2,4}}J_{{3,3}}+J_{{4,2}}J_{{2,4}}J_{{1,1}}+J_{{4,4}}J_{{3,1}}J_{{
1,3}}-J_{{4,4}}J_{{3,3}}J_{{2,2}}-J_{{4,4}}J_{{3,3}}J_{{1,1}}-J_{{4,4}
}J_{{2,2}}J_{{1,1}}+J_{{3,1}}J_{{1,3}}J_{{2,2}}-J_{{3,3}}J_{{2,2}}J_{{
1,1}},$

\begin{dmath*}
  C_3=  -J_{{4,4}}-J_{{3,3}}-J_{{2,2}}-J_{{1,1}}
\end{dmath*}
and
\begin{dmath*}
  C_4 = J_{{4,2}}J_{{2,4}}J_{{3,1}}J_{{1,3}}-J_{{4,2}}J_{{2,4}}J_{{3,3}}J_{{1,
		1}}-J_{{4,4}}J_{{3,1}}J_{{1,3}}J_{{2,2}}+J_{{4,4}}J_{{3,3}}J_{{2,2}}J_
{{1,1}}.
\end{dmath*}
Since $J_{{i,i}}<0$ for $i=1,2,3,4$, we deduce that $$C_3>0.$$
Moreover, straightforward computations lead $C_4>0$. Therefore the coexistence equilibrium is LAS whenever $C_1>0$ and $C_2>0$. This ends the proof.

\section{Bifurcation diagrams with respect to other parameters}\label{sec: appendix bif analysis other paras}

In the main text, we report on steady state behaviour under changes to anthropization in patch 1 ($\alpha_1$). We further performed a numerical bifurcation analysis in all other model parameters. We restricted our analysis to single parameter continuations only, fixing all parameters except the bifurcation parameter to $\alpha_1 = 0.2, \alpha_2=0.2, r_V = 3.0 $, $ r_A = 0.3 $, $ a_1 = 0.5 $, $ a_2 = 0.4 $, $ K_{V,1} = 1.0 $, $ K_{V,2} = 1.1 $, $ K_{A,1} = 0.5 $, $ K_{A,2} = 0.4 $, $ d_{V,12} = d_{V,21} = d_{A,12} = d_{A,21} = 0.2 $, $ c_V = 0.7 $, $ c_A = 0.3 $, $ b_1 = 1.0 $, $ b_2 = 1.1 $, $ \mu_{V,1} = 1.0 $, $ \mu_{V,2} = 1.1 $, $ \mu_{A,1} = 0.1 $, $ \mu_{A,2} = 0.11 $. These parameters were chosen to be close to the bifurcation between the coexistence and vector extinction steady state to ensure that the potential for bifurcations to check how each parameter affects the bifurcation, even if changes have a small quantitative impact. 

First, our numerical continuation confirmed that changing anthropization in patch 2 ($\alpha_2$) has the same qualitative impact as changing anthropization in patch 1 ($\alpha_1$) (c.f., \Cref{fig: contfigure}B and \Cref{fig: bif diags other paras}U). 

The effect of most other parameters (all except $n$) was monotonic. Increases in one of $b_1$, $c_A$, $c_V$, $K_{A,i}$, $r_A$, $r_V$ all cause a bifurcation from vector extinction to coexistence, with vector abundance increasing beyond the bifurcation point and animal densities increasing or remaining constant (if bifurcation parameter only affects vector dynamics) either side of the bifurcation (except $c_A$ which causes an increase of animal density in one patch and a decrease in the other because it represents the half-saturation constant of the migration term). Vice versa, increases in one of $b_2$, $a_i$, and any of the migration and mortality parameters cause a bifurcation from the coexistence state to the vector extinction state, with a monotonic decrease of the vector densities before reaching the bifurcation. Animal densities along these branches are either constant ($b_2$, $d_{V,12} = d_{V,21}$, $a_i$, $\mu_{V,i}$), decreasing ($\mu_{A,i}$) or have opposite effects on both animal densities ($d_{A,12} = d_{A,21}$, $d_{V,12} = d_{V,21} = d_{A,12} = d_{A,21}$. Changes in vector carrying capacities hat negligible impact on steady state densities, because vector growth was limited by animal abundance for the parameter values chosen. Finally, changes in the parameter $n$ describing the steepness of the migration strength functional response to changes in anthropization causes a bifurcation from vector extinction to coexistence. In contrast to other parameters, changes in $n$ causes the vector and animal abundance in one patch to increase, the animal abundance in the other patch to decrease, and the vector density in the other patch to increase initially before dropping again to near-zero values.

\begin{figure}
    \centering
    \includegraphics[width=\linewidth]{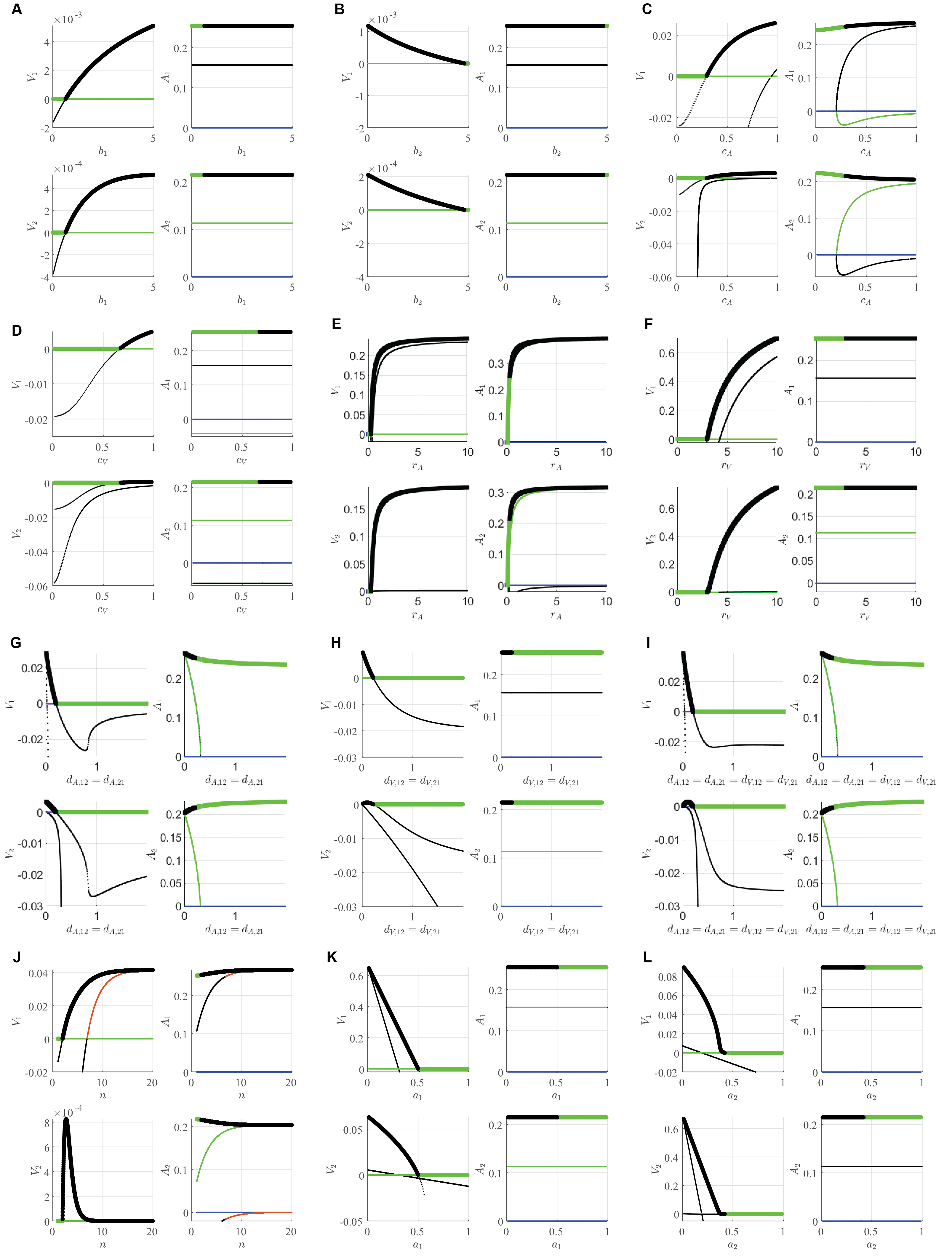}
    \caption{Bifurcation diagrams.}
\end{figure}
\addtocounter{figure}{-1}
\begin{figure}
    \centering
    \includegraphics[width=\linewidth]{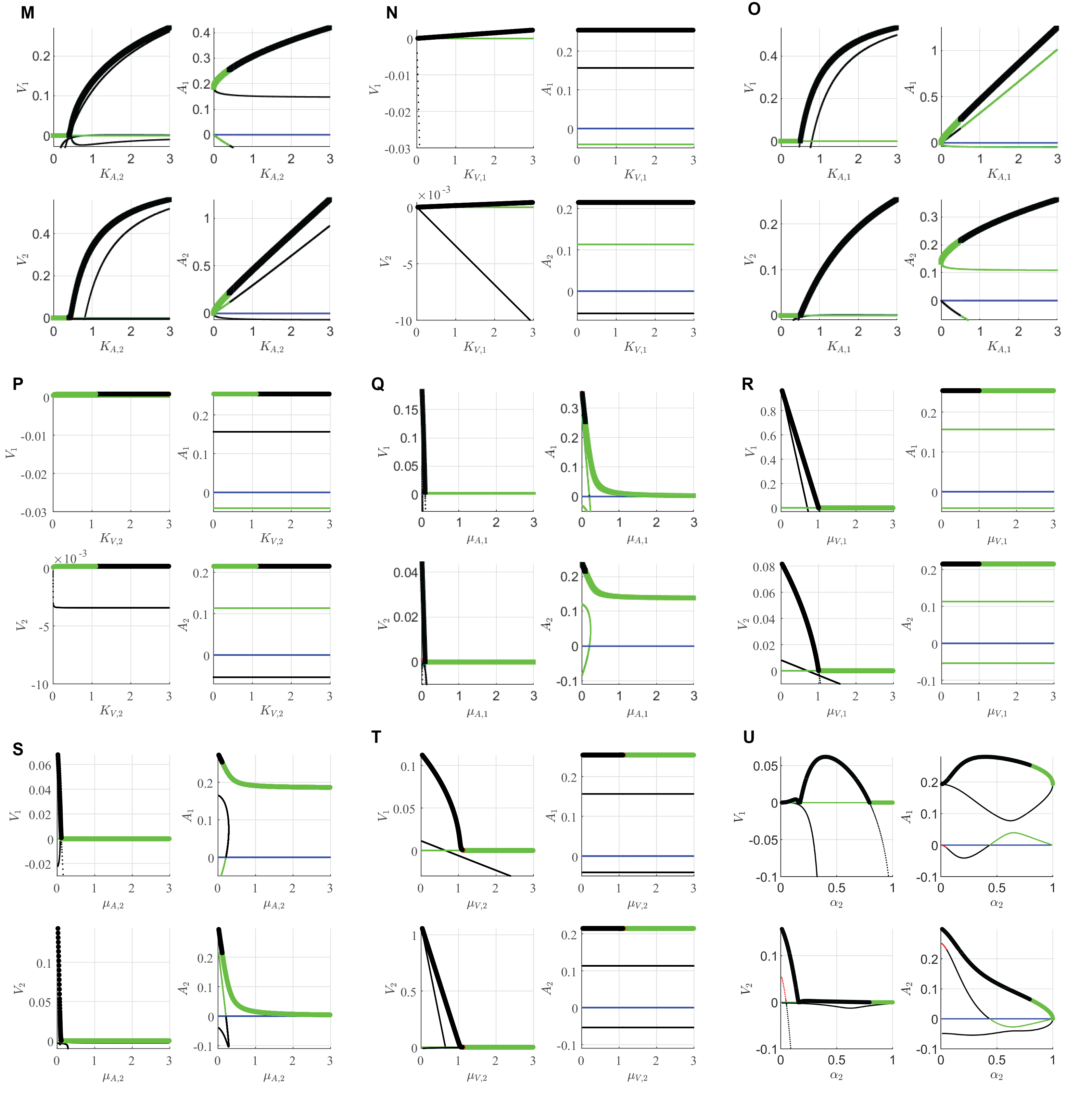}
    \caption{Bifurcation diagrams. Different panels show bifurcation diagrams for different bifurcation parameters. Thick curves represent stable solutions; thin curves represent unstable solutions. The colour of the curve distinguish between different equilibria: black curves are coexistence equilibria; green curves are vector extinction equilibria; blue curves are extinction equilibria; red curves are equilibria in which only $V_1$ goes extinct (never biologically relevant).}
    \label{fig: bif diags other paras}
\end{figure}

\end{document}